\documentclass[11pt]{article}
\usepackage{fullpage}
\usepackage[margin=0.5in]{geometry}
\usepackage{amssymb,amsfonts,amsmath,amsthm}
\usepackage{graphicx, xcolor}
\usepackage{caption}
\usepackage{subcaption}
\usepackage{cite}
\usepackage{hyperref}
\usepackage[capitalize]{cleveref}
\usepackage{authblk}
\usepackage{graphbox}

\usepackage{amsfonts}
\usepackage{amsmath}
\usepackage{amssymb}
\usepackage{amsthm}
\usepackage{bbm}
\usepackage{graphicx} 
\usepackage{mathrsfs}
\usepackage{float}
\usepackage{caption}
\usepackage{subcaption}
\usepackage{graphicx}

\usepackage{fullpage}
\usepackage[margin=0.5in]{geometry}
\usepackage{xcolor}
\usepackage{caption}
\usepackage{hyperref}
\usepackage[capitalize]{cleveref}
\usepackage{authblk}
\usepackage{graphbox}

\newtheorem{theorem}{Theorem}
\newtheorem{definition}{Definition}
\newtheorem{lemma}{Lemma}
\newtheorem{corollary}{Corollary}
\newtheorem{claim}{Claim}

\newcommand{\Graph}{\ensuremath{G}}
\newcommand{\Vertices}{\ensuremath{V}}
\newcommand{\Edges}{\ensuremath{E}}

\newcommand{\Degree}{\ensuremath{\mathrm{deg}}}
\newcommand{\InDegree}[1]{\ensuremath{\Degree^-(#1)}}
\newcommand{\OutDegree}[1]{\ensuremath{\Degree^+(#1)}}

\newcommand{\FixationProbability}{\ensuremath{\mathrm{fp}}}
\newcommand{\FixationProbabilityUnderNeutralEvolution}[1]{\ensuremath{\FixationProbability(#1)}}
\newcommand{\AbsorptionTime}{\ensuremath{\mathrm{AT}}}

\newcommand{\e}{\varepsilon}
\newcommand{\Set}[1]{\{#1\}}
\newcommand{\E}[1]{\ensuremath{\mathbb{E}\left[#1\right]}}
\newcommand{\Var}[1]{\ensuremath{\mathrm{Var}\left[#1\right]}}

\newcommand{\poly}{\text{poly}}

\let\leq=\leqslant 
\let\geq=\geqslant %
\let\le=\leqslant 
\let\ge=\geqslant %

\def\deg{\operatorname{deg}}

\def\indeg{\deg^-}
\def\outdeg{\deg^+}

\crefname{claim}{claim}{claims}

\def\T{\operatorname{T}}
\def\AT{\operatorname{AT}}
\def\FT{\T}
\def\ET{\operatorname{ExtT}}
\def\ATr{\AT_r}
\def\FTr{\FT_r}
\def\ETr{\ET_r}
\def\fp{\mathrm{fp}}
\def\fpr{\mathrm{fp}_r}
\def\twolayer{G}
\newcommand\pmin{\fp_{\rm min}}

\newcommand\fpmin{\pmin}

\def\good{balanced}
\def\short{fast}

\def\pup{p_r^+}
\def\pdown{p_r^-}
\def\bias{\gamma_r}
\def\eps{\varepsilon}

\theoremstyle{definition}

\def\eps{\varepsilon}

\def\deg{\operatorname{deg}}

\def\indeg{\deg^-}
\def\outdeg{\deg^+}

\def\T{\operatorname{T}}
\def\AT{\operatorname{AT}}
\def\FT{\T}
\def\ET{\operatorname{ExtT}}
\def\ATr{\AT_r}
\def\FTr{\FT_r}
\def\ETr{\ET_r}
\def\fp{\mathrm{fp}}
\def\fpr{\mathrm{fp}_r}
\def\twolayer{G}
\def\good{balanced}
\def\short{fast}
\def\SI{Appendices}

\title{
Fixation times on directed graphs%
}

\author[a]{David A. Brewster} 
\author[b,c]{Martin A. Nowak}
\author[b,d]{Josef Tkadlec}

\affil[a]{John A. Paulson School of Engineering and Applied Sciences, Harvard University, Cambridge, MA~02138,~USA}
\affil[b]{Department of Mathematics, Harvard University, Cambridge, MA 02138, USA}
\affil[c]{Department of Organismic and Evolutionary Biology, Harvard University, Cambridge MA~02138,~USA}
\affil[d]{Computer Science Institute, Charles University, Prague, Czech Republic}
\date{}

\begin{document}

\maketitle

%
%






\begin{abstract}
Computing the rate of evolution in spatially structured populations is difficult.
A key quantity is the fixation time of a single mutant with relative reproduction rate $r$ which invades a population of residents.
We say that the fixation time is ``\short'' if it is at most a polynomial function in terms of the population size $N$.
Here we study fixation times of advantageous mutants ($r>1$) and neutral mutants ($r=1$) on \textit{directed} graphs,
which are those graphs that have at least some one-way connections.
We obtain three main results.
First, we prove that for any directed graph the fixation time is \short, provided that $r$ is sufficiently large.
Second, we construct an efficient algorithm that gives an upper bound for the fixation time for any graph and any $r\ge 1$.
Third, we identify a broad class of directed graphs with \short{} fixation times for any $r\ge 1$.
This class includes previously studied amplifiers of selection, such as Superstars and Metafunnels.
We also show that on some graphs the fixation time is not a monotonically declining function of $r$; in particular, neutral fixation can occur faster than fixation for small selective advantages. 
\end{abstract}


\section{Introduction}
Evolution is a stochastic process that acts on populations of reproducing individuals.
Two main driving forces of evolutionary dynamics are mutation and selection~\cite{kimura1968evolutionary,moran1958random,ewens2004mathematical}.
Mutation generates new variants and selection prunes them.
When new mutations are sufficiently rare, the evolutionary dynamics are characterized by the fate of a single new mutant.
The mutant can either take over the whole population or become extinct.
Even when the mutation grants its bearer a relative fitness advantage $r\ge 1$,
it might still go extinct due to random fluctuations~\cite{durrett1994importance}.
Two key parameters that quantify the fate of the newly occurring mutation are the fixation probability and the fixation time~\cite{whitlock2003fixation,slatkin1981fixation}.
Here we study the effect of spatial structure on those quantities.

Spatial models have a long history of investigation in ecology
\cite{Bonachela_Pringle_Sheffer_Coverdale_Guyton_Caylor_Levin_Tarnita_2015,Tilman_May_Lehman_Nowak_1994,Comins_Hassell_May_1992,Hassell_Comins_Mayt_1991,Levin_1992,Levin_1976,tilman1997spatial}, population dynamics \cite{durrett1994importance}, population genetics \cite{Barton_1993,Fisher_Ford_1950,Maruyama_1970,Nagylaki_Lucier_1980,S_1932,slatkin1981fixation,Wright_1931}, evolutionary game theory \cite{Allen_Lippner_Chen_Fotouhi_Momeni_Yau_Nowak_2017,Hauert_Doebeli_2004,Nakamaru_Matsuda_Iwasa_1997,Nowak_May_1992,Ohtsuki_Hauert_Lieberman_Nowak_2006}, infection dynamics \cite{Grenfell_Bjørnstad_Kappey_2001,Lloyd_May_1996,May_Lloyd_2001,May_Nowak_1994} and cancer evolution \cite{Noble_Burri_Le_Sueur_Lemant_Viossat_Kather_Beerenwinkel_2022,Waclaw_Bozic_Pittman_Hruban_Vogelstein_Nowak_2015}. The classical investigation in population genetics includes the debate between Fisher and Wright \cite{Fisher_Ford_1950} and hybrid zones \cite{Barton_1979}. A biological population can be structured in the sense of geographical distribution, age structure, or specific interaction patterns. Human populations structure is often studied in terms of social networks \cite{Mitchell_1974}.

Spatial structure has profound effects on both the fixation probability and the fixation time.
Those effects are studied within the framework of
Evolutionary Graph Theory~\cite{lieberman2005evolutionary,diaz2021survey,donnelly1983finite}.
There, individuals are represented as nodes of a graph (network).
The edges (connections) of the graph represent the migration patterns of offspring.
The edges can be one-way or two-way.
Graphs can represent the well-mixed population, spatial lattices, island sub-populations, or arbitrary complex spatial structures.
Those directed graphs can arise by the flow from upstream to downstream demes in meta-populations, by cellular differentiation in somatic evolution, or in age structured populations.
Also in human social networks many interactions are one way---such as from influencer to follower or from teacher to learner. 

Previous research investigated population structures with various effects on fixation probability and time \cite{hindersin2015most,broom2010evolutionary,broom2010two,monk2020wald,monk2021martingales}.
For example, isothermal graphs have both the same fixation probability
and the same distribution of mutant population size changes over time as the well-mixed population \cite{lieberman2005evolutionary,monk2020wald}.
Suppressors of selection reduce the fixation probability of advantageous mutants \cite{nowak2003linear}, and amplifiers of selection enhance the fixation probability of advantageous mutants \cite{nowak2006evolutionary}.
Amplifiers are population structures that could potentially accelerate the evolutionary search \cite{adlam2015amplifiers}.
Known classes of amplifiers include families such as Stars \cite{hadjichrysanthou2011evolutionary,monk2014martingales,chalub2014asymptotic}, Comets \cite{pavlogiannis2017amplification}, Superstars \cite{diaz2013fixation,jamieson2015fixation}, or Megastars \cite{galanis2017amplifiers}.

Interestingly, the amplification typically comes at a cost of increasing the fixation time \cite{tkadlec2019population}, sometimes substantially \cite{pavlogiannis2018construction}.
This is problematic, since when fixation times are extremely long, fixation is not a relevant event any more,
and thus the fixation probability alone is not the most representative quantity \cite{tkadlec2021fast,sharma2022suppressors}.
It is therefore paramount to understand how the population structure affects the fixation time and, in particular, what are the population structures for which the fixation time is ``reasonably \short{}.''


Borrowing standard concepts from computer science \cite{cormen2022introduction}, in this work we say that
fixation time is \textit{\short{}} if the fixation time is (at most) polynomial in terms of the population size $N$.
Otherwise we say that the fixation time is \textit{long}, and the corresponding population structure is \textit{slow}.
Two important known results are: (i) for all \textit{undirected} graphs the fixation time is \short{} \cite{diaz2014approximating,goldberg2020phase};
and (ii) if some edges are one-way (if the graph is directed),
then the fixation time can be exponentially long \cite{diaz2016absorption}.
The latter result has an important negative consequence:
when the fixation time is exponentially long,
we know no tool to efficiently simulate the process.
Therefore, computing or approximating any relevant quantities for realistic population sizes is in practice infeasible.

In this work, we present three positive results that concern fixation times on directed graphs (where some or all edges are one-way).
First, we prove that for any directed graph the fixation time is \short{},
provided that the mutant fitness advantage $r$ is sufficiently large.
Second, we devise an efficient algorithm that gives an upper bound on the
fixation time, for any graph and any $r\ge 1$.
The bound can be used to estimate how long one needs to run the simulations until they terminate.
Third, we identify a broad class of directed graphs for which the fixation
times are \short{} for any $r\ge 1$.
This class includes many previously studied amplifiers of selection, such as Superstars and Metafunnels.
To conclude, we discuss important algorithmic consequences that enable efficient computational exploration of various properties of directed graphs.


\section{Model}
In this section we define the notions we use later, such as
the population structure (represented by a graph),
the evolutionary dynamics (Moran Birth-death process),
and the key quantities (fixation probability and fixation time).

\subsection{Population structure}
The spatial population structure is represented by a graph (network) $\Graph=(\Vertices,\Edges)$, where $\Vertices$ is the set of $N$ nodes (vertices) labeled $v_1,\dots,v_N$ and $\Edges$ is the set of directed one-way edges (links) connecting pairs of different nodes.
A two-way connection between nodes $u$ and $v$ is represented by two one-way edges $u\to v$ and $v\to u$.
We assume that the graph is strongly connected.
For any node $v$, the number of edges incoming to $v$ is called the indegree (denoted $\indeg(v)$), and the number of outgoing edges is called the outdegree (denoted $\outdeg(v)$).
When the two numbers coincide, we call them the degree (denoted $\deg(v)$).

\subsection{Graph classes}
We say that a graph is \textit{undirected} if for every edge $u\to v$ there is also an edge $v\to u$ in the opposite direction.
Otherwise we say that a graph is \textit{directed}.
We say that a graph is \textit{regular} if all nodes have the same degree, that is, there exists a number $d$ such that $\indeg(v)=\outdeg(v)=d$ for all nodes $v\in \Vertices$.
We say that a graph $\Graph$ is \textit{Eulerian} (also known as a circulation) if $\indeg(v)=\outdeg(v)$ for each node $v$.
Finally, in this work we say that a graph is \textit{\good{}} if an equality
\[\frac1{\outdeg(v)} \cdot \sum_{w\in \Vertices\colon v\to w \in \Edges} \frac1{\indeg(w)}
=
\frac1{\indeg(v)}\cdot \sum_{u\in \Vertices\colon u\to v\in \Edges} \frac1{\outdeg(u)}
\]
holds for all nodes $v$.
Here the left-hand side represents the average indegree of the successors of $v$, whereas the right-hand side represents the average outdegree of the predecessors of $v$.
It is straightforward to check that the class of \good{} graphs includes the regular graphs and the undirected graphs, as well as other graph classes such as Superstars or Metafunnels\cite{lieberman2005evolutionary}, see \SI.
Below we will prove that the fixation times on all \good{} graphs are \short{} for any $r\ge 1$.

\subsection{Moran Bd process}
To model the evolutionary dynamics we consider the standard Moran Birth-death process.
Each node of the graph is occupied by a single individual.
Initially, some individuals are wild-type residents with normalized fitness equal to~$1$, and some individuals are mutants with relative fitness advantage $r\ge 1$.
Given a graph $\Graph$ and a relative fitness advantage $r\ge 1$, Moran Birth-death process is a discrete-time stochastic process, where in each step:
\begin{enumerate}
\item First (Birth), we select an individual with probability proportional to its fitness. Suppose we selected node~$u$.
\item Second (death), we select an outgoing neighbor of $u$ uniformly at random. Suppose we selected node~$v$.
\item Finally (update), we replace the individual at node $v$ by a copy of individual at node~$u$.
\end{enumerate}
At each time-step, the current \textit{configuration} is the subset of nodes occupied by mutants.
Since the graph is strongly connected, almost surely we eventually obtain a configuration where either all nodes are mutants (we say that mutants \textit{fixed}), or all nodes are residents (we say that mutants \textit{went extinct})~\cite{diaz2014approximating}.

\subsection{Fixation probability and fixation time}
The key quantities that we consider in this work are fixation probability and fixation time.

Given a graph $\Graph$, a mutant fitness advantage $r\ge 1$, and a current configuration $S\subseteq\Vertices$ of nodes occupied by mutants, the \textit{fixation probability} $\fp_r(\Graph,S)$ is the probability that starting from $S$, the mutants eventually fix (as opposed to going extinct).
Morever, we define an auxiliary quantity $\pmin$ that turns out to be useful later in our results.
Formally, given a graph $\Graph$ and $r=1$, for $i=1,\dots,N$ denote by $\fp^{(i)}(\Graph)=\fp_{r=1}(\Graph,\{v_i\})$ the fixation probability of a single neutral mutant who initially appears at node $v_i$. We define $\pmin(\Graph)=\min_i \fp^{(i)}(\Graph)$ to be the smallest of those $N$ fixation probabilities.

To measure the duration of the process until fixation (or extinction) occurs, different notions are used.
The \textit{(expected) absorption time} $\ATr(\Graph,S)$ is the expected number of steps of the Moran Birth-death process until the process terminates, regardless of what is the outcome (mutant fixation or extinction).
In contrast, \textit{(expected) fixation time} $\FTr(\Graph,S)$ is the expected number of steps averaged over only those evolutionary trajectories that terminate with mutant fixation.
Similarly, one can define the \textit{(expected) extinction time} $\ETr(\Graph,S)$ averaging over only those trajectories that terminate with the mutant going extinct.
By linearity of expectation, the three quantities are related as $\ATr(\Graph,S) = \fp_r(\Graph,S)\cdot \FTr(\Graph,S) + (1-\fp_r(\Graph,S))\cdot\ETr(\Graph,S)$.
Note that in this work, absorption time, fixation time, and extinction time are mean times to absorption, making them scalar values rather than random variables.
Information about the random variable can be recovered from its expectation using concentration bounds such as Markov's inequality~\cite{diaz2014approximating}.
Our objective in this work is to provide upper bounds on the absorption time and on the fixation time.
To that end, given a graph $\Graph$ and a mutant fitness advantage $r\ge 1$, let $\FTr(\Graph)=\max_{S\subseteq\Vertices, S\neq\emptyset} \FTr(\Graph,S)$ be the largest fixation time among all possible initial configurations. 
In the limit of strong selection $r\to\infty$ we also define $\T_\infty(\Graph)=\lim_{r\to\infty} T_r(\Graph)$.
This regime is called the ecological scenario~\cite{ibsen2015computational} and corresponds to new invasive species populating an existing ecosystem.

\subsection{Asymptotic notation}
We say a function $f(N)$ is (at most) \textit{polynomial} if there exists a positive constant $c$ such that $f(N)\le N^c$ for all large enough $N$.
Examples of polynomial functions are $f(N)=\frac12N(N+1)$ and $f_2(N)=10\cdot N\log N$, whereas functions such as $g(N)=1.1^N$ and $g_2(N)=2^{\sqrt N}$ are not polynomial, since they grow too quickly.
In computer science, problems that can be solved using polynomially many elementary computations are considered tractable.
In alignment with that, given a population structure $\Graph$ with $N$ nodes, we say that fixation time is \textit{\short{}} if it is
at most polynomial in terms of the population size $N$.

\section{Results}
We present three main types of results.


\subsection{Fixation time is \short{} when selection advantage is strong enough}

As our first main result, we prove that the fixation time on any directed graph is \short{}, provided that the mutant fitness advantage $r$ is large enough.

As an illustration, for every $N=4k$ we consider a graph $\twolayer_N$ depicted in \cref{fig:transient}(a).
It consists of four columns of $k$ nodes each.
The grey edges within the yellow region are two-way.
The black one-way edges point from the side columns to the middle columns.
When mutants initially occupy the left part of the graph, the only way forward for them is to progress upward through the third column.
But while there, mutants are under an increased pressure due to the resident nodes in the rightmost column.
The same applies to residents. They can only make progress by climbing upward through the second column, but there they are under pressure due to mutants in the leftmost column. 
As a consequence, the fixation time crucially depends on $r$.
When $r= 1.1$, \cref{fig:transient}(b) shows that the fixation time scales exponentially in $N$ (that is, it is long).
In contrast, in the limit of large $r$ 
the fixation time is less than $N^2$, that is, it is \short. 


\begin{figure}[h]
  \centering
   \includegraphics[width=\linewidth]{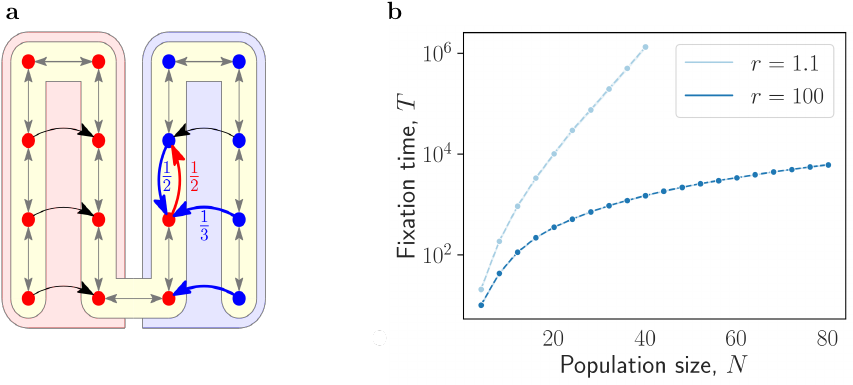}
\caption{
\textbf{Long and \short{} fixation times on a four-column graph $\twolayer_N$.}
\textbf{a,} 
For $N=4k$, a graph $\twolayer_N$ consists of four columns of $k$ nodes each (here $k=4$ and $N=16$).
The grey edges within the yellow region are two-way, the black edges going from the side columns to the corresponding vertices of the middle column are one-way.
Initially mutants occupy the left half $L$ and residents occupy the right half.
As mutants (red nodes) spread upward through the third column, they can propagate along only one edge (red), whereas residents (blue nodes) fight back along multiple edges (blue).
\textbf{b,} The timescale to fixation crucially depends on the mutant fitness advantage $r$. When $r=1.1$, the fixation time $\FTr(\twolayer_N,L)$ is exponential in $N$, whereas when $r=100$ it is polynomial. Each data point is an average over at least $10^3$ simulations.
}
\label{fig:transient}
\end{figure}

In general, we can prove the following result about an arbitrary population structure.
\begin{theorem}\label{thm:r-infty} Let $\Graph_N$ be an arbitrary graph on $N$ nodes.
Suppose that $r\ge N^2$.
Then $\ATr(\Graph_N)\leq 2N^3$ and $\FTr(\Graph_N)\le 3N^3$.
\end{theorem}

\cref{thm:r-infty} implies that while the fixation time on certain graphs can be long for some values of $r$, this effect is inevitably transient, and the fixation time becomes \short{} once $r$ exceeds a certain threshold.
The intuition behind the proof is that if $r$ is large enough, the size of the mutant subpopulation is always more likely to increase than to decrease,
regardless of which nodes are currently occupied by mutants.
Thus, the evolutionary process can be mapped to a random walk with a constant positive bias.
It is known that such biased random walks absorb polynomially quickly.
See \SI{} for details.

An attractive feature of \cref{thm:r-infty} is that it applies to all directed graphs. 
An obvious limitation is that the condition $r\ge N^2$ is unrealistic, except possibly for the regime $r\to\infty$ that has been studied under the name ecological scenario~\cite{ibsen2015computational}.
Therefore, as our second result, we considerably relax this condition for graphs with certain structural features.
A directed graph is said to be Eulerian (also called a \textit{circulation}) if each node has the same indegree as outdegree.
In that case, we refer to the number $\indeg(v)=\outdeg(v)$ simply as a \textit{degree} of node $v$.

\begin{theorem}\label{thm:eulerian}
    Let $\Graph_N$ be an Eulerian graph on $N$ nodes
    with smallest degree $\delta$ and largest degree $\Delta$.
    Suppose that $r\geq \frac{\Delta}{\delta}\cdot(1+\e)$
    for some $\e>0$.
    Then
    $\ATr(\Graph_N) \leq \frac{2+\e}{\e}\cdot N^3$
    and
    $\FTr(\Graph_N)\leq \frac{(1+\e)(2+\e)}{\e^2}\cdot  N^{3}$.
\end{theorem}

To illustrate~\cref{thm:eulerian} we point out two special cases (for the full proof see~\SI).

First, consider any regular graph $\Graph_N$, that is, a graph where all nodes have the same indegree and outdegree equal to $d$.
Then, the graph is Eulerian and we have $d=\Delta=\delta$, and thus \cref{thm:eulerian} implies that $\FTr(\Graph_N)$ is at most a polynomial in $N$ and $d$ for any $r\ge 1$. In other words, fixation time on any regular graph is \short{} for any $r\ge 1$ (we note that this result is known\cite{diaz2016absorption}).

Second, consider an Eulerian graph that is ``close to being regular'', in the sense that each node has degree either 4 or 5.
An example of such a graph is a square lattice with several additional long-range connections.
Then, \cref{thm:eulerian} implies that the fixation time is \short{} for every $r> 5/4=1.25$.

\subsection{Fixation time for
small selective advantage}
The above results show that for any fixed graph $\Graph$, the fixation time is \short{} when $r$ is sufficiently large.
It is natural to hope that perhaps for any fixed graph $\Graph$ the fixation time is a monotonically decreasing function of $r$ for $r\ge 1$.
However, this is not the case, as shown in \cref{fig:t-nonmonotone}.

\begin{figure}[h]
  \centering
\includegraphics[width=\linewidth]{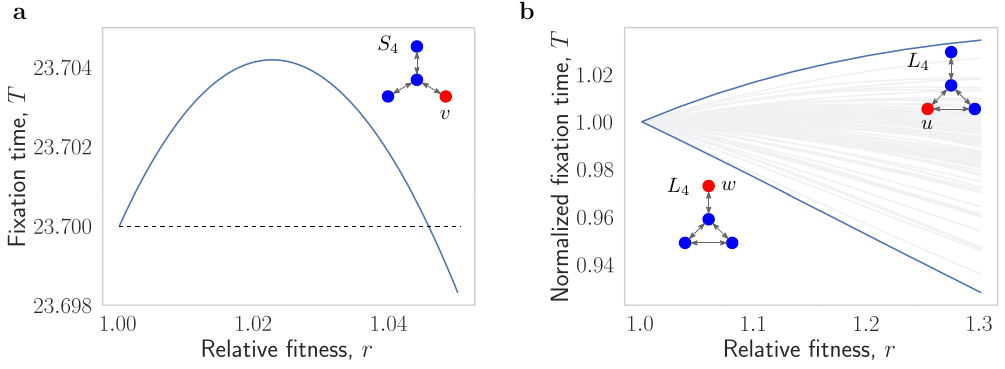}
\caption{
\textbf{Fixation time is not monotone in $r$.}
\textbf{a,} In an (undirected) star graph $S_4$ on 4 nodes, one node (center) is connected to three leaf nodes by two-way edges. When the initial mutant appears at a leaf $v$, the fixation time $\FTr(S_4,\{v\})$ increases as $r$ increases from $r=1$ to roughly $r=1.023$. Then it starts to decrease.
\textbf{b,} Normalized fixation time $\FTr(\Graph,\{v\})/\T_{r=1}(\Graph,\{v\})$ as a function of $r\in[1,1.3]$, for all 83 strongly connected graphs $\Graph$ with 4 nodes, and all four possible mutant starting nodes $v$. 
As $r$ increases, the fixation time goes up for 182 of the $4\cdot 83=332$ possible initial conditions.
The increase is most pronounced for the so-called \textit{lollipop} graph $L_4$ and a starting node $u$.
In contrast, for the same lollipop graph and a different starting node $w$, the fixation time decreases the fastest.
}
\label{fig:t-nonmonotone}
\end{figure}

Briefly speaking, the effect responsible for the increase in fixation time when $r=1+\eps$ is that by increasing the mutant fitness advantage, certain evolutionary trajectories that used to lead to mutant extinction instead lead to mutant fixation.
Since those ``newly fixating'' trajectories might generally take relatively long to fix, the average length of the fixating trajectories can go up. Similarly, the absorption time can also go up as we increase $r$. Those findings are in alignment with the stochastic slowdown phenomenon~\cite{Altrock_Gokhale_Traulsen_2010}.


Despite the lack of monotonicity, we can show that the fixation time cannot go up too much as we increase~$r$.
Recall that $\pmin(\Graph_N)=\min\{\fp_{r=1}(\Graph_N,\{v\}) \mid v\in \Graph_N\}$ denotes the fixation probability under neutral drift ($r=1$), when the initial mutant appears at a node $v$ with the smallest fixation probability.
Note that for any graph with $N$ nodes we have $\pmin(\Graph_N)\le 1/N$, but $\pmin(\Graph_N)$ could in general be substantially smaller than $1/N$.
Finally, the quantity $\pmin(\Graph_N)$ can be computed efficiently by solving a linear system of $N$ equations \cite{donnelly1983finite,broom2010two,maciejewski2014reproductive}.

We can now state our second main result.
\begin{theorem}\label{thm:tr-t1}
Fix a graph $\Graph_N$ and $r\ge 1$.
Then $\FTr(\Graph_N)\le
\frac {N^6}{(\pmin(\Graph_N))^4}.$
\end{theorem}

We note that \cref{thm:tr-t1} yields an efficiently computable upper bound on $\FTr(\Graph_N)$.
In the next section we elaborate on the computational consequences of this result.
In the rest of this section, we give a brief intuition behind the proof of~\cref{thm:tr-t1}.

The proof relies on two ingredients. 
First, instead of considering the process with mutant fitness advantage $r$, we consider the neutral process that corresponds to $r=1$.
There, using a martingale argument we show that the fixation time $\T_{r=1}(\Graph_N)$ can be bounded from above in terms of the quantity $\pmin$.
The intuition is that as long as all fixation probabilities are non-negligible,
all active steps of the stochastic process have substantial magnitude either towards fixation or towards extinction.
As a consequence, we are able to argue that either fixation or extinction will occur after not too many steps.
All in all, this yields an upper bound on fixation time $\T_{r=1}(\Graph_N)$ of the neutral process in terms of the quantity $\pmin$. See 
\SI{} for details.

As our second ingredient, we translate the bound on $\T_{r=1}(\Graph_N)$ into a bound on $\FTr(\Graph_N)$ for any $r\ge 1$.
We note that, as indicated in \cref{fig:t-nonmonotone}, for a fixed graph $\Graph_N$ the fixation time is in general not a monotonically decreasing function of $r$.
Nevertheless, the continuous versions of two processes can be coupled in a certain specific way, which allows us to argue that 
while $\FTr(\Graph_N)$ can be somewhat larger than $\T_{r=1}(\Graph_N)$, it cannot be substantially larger.
In this step, we again use the quantity $\pmin$. See 
\SI{} for details.

\subsection{Fixation time is \short{} when the graph is \good}
As noted above, \cref{thm:tr-t1} provides an upper bound on the fixation time for any graph $\Graph_N$ and any mutant fitness advantage $r\ge 1$, in terms of the quantity $\pmin(\Graph_N)$.
We have $0\le \pmin(\Graph_N)\le 1/N$.
When the quantity $\pmin(\Graph_N)$ is exponentially small, the upper bound from \cref{thm:tr-t1} becomes exponentially large, and thus not particularly interesting.
However, for many graphs the quantity $\pmin(\Graph_N)$ turns out to be much larger, namely inversely proportional to a polynomial in $N$.
In those cases, \cref{thm:tr-t1} implies that the fixation time $\FTr(\Graph_N)$ is \short{} for any $r\ge 1$.

In particular, as our third main result we prove that this occurs for a broad class of graphs which we call \good{} graphs.
Formally, we say that a graph $\Graph_N$ is \textit{\good{}} if an equality
\[
\frac1{\indeg(v)}\cdot \sum_{u\colon u\to v\in \Edges} \frac1{\outdeg(u)}
=
\frac1{\outdeg(v)} \cdot \sum_{w\colon v\to w \in \Edges} \frac1{\indeg(w)}
\]
holds for all nodes $v$.
Here the left-hand side represents the average outdegree of the predecessors of $v$, whereas 
the right-hand side represents the average indegree of the successors of $v$.
We note that the family of \good{} graphs includes many families of graphs studied in the context of Moran process in the existing literature, such as the undirected graphs \cite{diaz2014approximating}, regular (possibly directed) graphs \cite{diaz2016absorption}, Superstars and Metafunnels \cite{lieberman2005evolutionary},
Megastars~\cite{galanis2017amplifiers,monk2018martingales},
cyclic complete multipartite graphs~\cite{Monk_van_Schaik_2022}, or directed fans \cite{allen2020transient}, see \cref{fig:graph-classes}.

\begin{figure}[h]
  \centering
  \includegraphics[align=c,width=\linewidth]{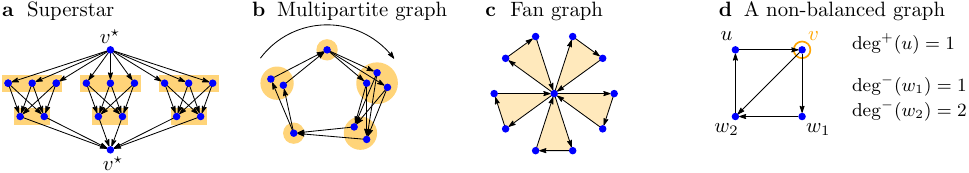}
\caption{
\textbf{Types of \good{} graphs.}
The class of \good{} graphs includes the following families of graphs studied in the context of Moran process in the existing literature.
\textbf{a,} Superstars~\cite{lieberman2005evolutionary} are the first proposed strong amplifiers of selection.
\textbf{b,} Complete multipartite graphs~\cite{Monk_van_Schaik_2022} are a rare example of high-dimensional graphs for which the fixation probability of advantageous mutants can be expressed using an explicit formula.
\textbf{c,} A certain form of Fan graphs~\cite{allen2020transient} (with weighted and undirected edges) constitutes the strongest currently known amplifiers of selection under death-Birth updating.
\cref{thm:good} implies that the fixation time on all those graphs is \short{} for all $r\ge 1$.
\textbf{d,} Not all graphs are balanced. For example, here for the highlighted node $v$ the left hand side is $1$, whereas the right-hand side is $\frac12(1/1+1/2)=0.75\ne 1$.
}
\label{fig:graph-classes}
\end{figure}

We have the following theorem.

\begin{theorem}\label{thm:good}
Let $\Graph_N$ be a \good{} graph. Then:
\begin{enumerate}
\item $\fp_{r=1}(\Graph_N,u)=\frac{1/\indeg(u)}{\sum_{v\in \Vertices} 1/\indeg(v)}\ge 1/N^2$ for any node $u$.
\item $\FTr(\Graph_N)\le N^{14}$ for any $r\ge 1$.
\end{enumerate}
\end{theorem}

\cref{thm:good} implies that the fixation time on all \good{} graphs is \short{} for all $r\ge 1$.
Similarly, we can prove that it is \short{} for Megastars \cite{galanis2017amplifiers} assuming $r\ge 1$ (see \SI).

The proof of the first part of \cref{thm:good} relies on the fact that in the neutral case $r=1$ the fixation probability is additive.
This allows us to reduce the size of the linear system that describes the underlying Markov chain from $2^N$ equations to $N$ equations.
For \good{} graphs, this system takes a special form that admits an explicit solution. The second part then follows directly from \cref{thm:tr-t1}. See \SI{} for details.

We note that the second part of \cref{thm:tr-t1} has an important computational consequence.
Since the fixation time on any \good{} graph is bounded from above for any $r\ge 1$, 
individual-based simulations of the evolutionary process are guaranteed to terminate quickly with high probability \cite{diaz2014approximating}.
Any relevant quantities of interest, such as the fixation probability of the mutant with $r\ge 1$, can thus be efficiently approximated to arbitrary precision.
In particular, \cref{thm:tr-t1} yields a fully-polynomial randomized approximation scheme (FPRAS) for the fixation probability on \good{} graphs with any $r\ge 1$.

\begin{theorem}\label{thm:fpras} There is a FPRAS for fixation probability on \good{} graphs for any $r\ge 1$.
\end{theorem}

We note that \cref{thm:fpras} applies also to any (not necessarily \good{}) graph $\Graph_N$, provided that the quantity $\pmin(\Graph_N)$ is inversely proportional to a polynomial. This is the case for instance for Megastars.
See \SI{} for details.
Moreover, when $\pmin(\Graph_N)$ is smaller than that, \cref{thm:tr-t1} still gives an explicit, efficiently computable upper bound on the fixation time that can be used to bound the running time of any individual-based simulations.


\subsection{Computational experiments}
Finally, to further illustrate the scope of our results we run several computational experiments on graphs with small population size $N$.
We use \texttt{nauty} \cite{mckey1990nauty} to perform such
enumerations.
Since already for $N=6$ there are more than one million non-isomorphic strongly connected directed graphs, 
we consider $N= 5$.
For each of the 5048 graphs with $N= 5$ we compute the fixation time and the fixation probability under uniform initialization by solving the underlying Markov chain using numerical methods (\cref{fig:1M}).

\begin{figure}[h]
  \centering
   \includegraphics[width=\linewidth]{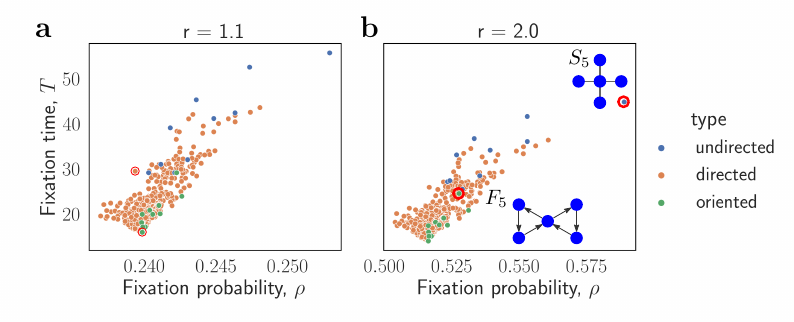}
\caption{
\textbf{} 
Fixation time and fixation probability of a single mutant under uniform initialization for all 5048 graphs with $N= 5$ nodes, for \textbf{a,} $r=1.1$ and \textbf{b,} $r=2$. Each graph is represented as a colored dot.
The undirected graphs (with all edges two-way) are labeled in blue.
The oriented graphs (with no edges two-way) are labeled in green.
All other directed graphs are labeled in orange.
The slowest graph is the (undirected) Star graph $S_5$.
Among the oriented graphs, the slowest graph is the Fan graph $F_5$.
}
\label{fig:1M}
\end{figure}

The slowest graph is the (undirected) Star graph.
Note that when $N$ is large the fixation time on a Star graph is known to be proportional to roughly $N^2$ \cite{tkadlec2019population}.

Among the oriented graphs, the slowest are variants of either a fan graph $F_N$, or a vortex graph $V_N$.
Since both the fan graphs and the vortex graphs belong to the class of \good{} graphs, the fixation time on those graphs is \short{} for any population size $N$ and any mutant fitness advantage $r\ge 1$ due to \cref{thm:good} (see \cref{fig:fans} for empirical support).
The fixation time appears to be proportional to roughly $N^2$ (see \cref{sec:oriented}).

\begin{figure}[h]
  \centering
 \includegraphics[width=\linewidth]{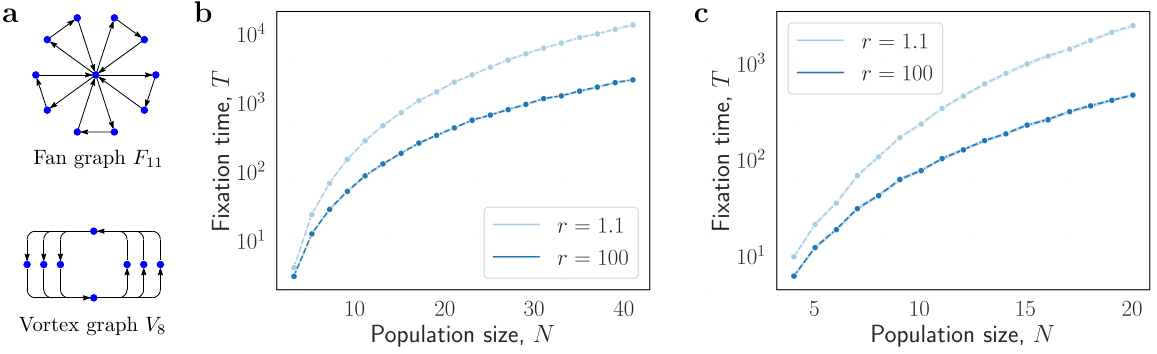}
   
\caption{
\textbf{Fixation time on slow oriented graphs.}
\textbf{a,} The Fan graph with $k$ blades has $N=2k+1$ nodes and $3k$ one-way edges (here $k=5$ which yields $N=11$). The Vortex graph with batch size $k$ has $N=2k+2$ nodes and $4k$ edges (here $k=3$ which yields $N=8$).
\textbf{b-c,} For both the Fan graphs and the Vortex graphs the fixation time scales roughly as $N^2$, both for $r=1.1$ and $r=100$. (Each data point is an average over 1000 simulations.)
}
\label{fig:fans}
\end{figure}

Together, those results suggest that even though directed graphs with exponentially long fixation times do exist, in practice most small directed graphs reach fixation reasonably quickly.


\section{Discussion}
A decade ago, a foundational work by Diaz et.~al.~showed that fixation time on any undirected population structure is \short{}~\cite{diaz2014approximating}.
This result enabled extensive computational exploration of the landscape of all undirected graphs that later lead to several inspiring research outputs~\cite{hindersin2015most,moller2019exploring,Allen_Lippner_Chen_Fotouhi_Momeni_Yau_Nowak_2017,goldberg2019asymptotically,tkadlec2021fast}.
It is our hope that by enabling computational exploration of population structures with some (or all) one-way connections, this work will serve the same purpose.

Studying the evolutionary dynamics in spatially structured populations is notoriously hard.
In this work we consider one of the simplest possible dynamics, namely the classic Moran Birth-death process, and we analyze the fixation time of a newly occurring mutant.
When the fixation time is exponentially long, the process is expensive to simulate, and moreover various commonly studied quantities such as fixation probability are largely irrelevant.
It is thus paramount to delineate settings in which the fixation time is ``relatively \short{}'', as opposed to being ``exceedingly long.''

It is known~\cite{diaz2014approximating} that the fixation time is fast, provided that all interactions among the individuals are two-way.
However, many relevant population structures include one-way interactions.
In meta-population dynamics, an upstream deme could seed a downstream deme. In somatic evolution, cellular differentiation could be irreversible. In an age structured population there is the one way arrow of time.

In this work, we therefore consider spatial structures in which some (or all) interactions are one-way. It is known that on such structures the fixation time can be exceedingly long \cite{diaz2016absorption}.
Nevertheless, here we present three results which indicate that fixation times on spatial structures with one-way connections are often \short{}.

First, we prove that on any population structure the fixation time is \short{}, provided that the mutant fitness advantage $r$ exceeds a certain threshold value $r^\star$ (see \cref{thm:r-infty}).
In the special case when the population structure is represented by a regular graph, the threshold value simplifies to $r^\star =1$, and we recover a known result that the fixation time on regular graphs is short for all $r\ge 1$ \cite{diaz2016absorption}.
As another corollary, for any Eulerian graph whose degrees are sandwiched between $\delta$ and $\Delta$ we can set $r^\star=\Delta/\delta$ (see \cref{thm:eulerian}).

Second, somewhat counter-intuitively we show that fixation time sometimes goes up as we increase $r$.
That is, on certain spatial structures fixation of a neutral mutant occurs faster than fixation of a mutant with a small selective advantage. This effect is called stochastic slowdown~\cite{Altrock_Gokhale_Traulsen_2010}.
We show that the magnitude of the slowdown can be bounded.
In particular, in the spirit of parametrized complexity~\cite{downey2012parameterized,cygan2015parameterized}, given a graph structure $\Graph_N$ we define a certain efficiently computable quantity $\pmin(\Graph_N)$,
and we bound the fixation time for any $r\ge 1$ from above using $\pmin(\Graph_N)$ and $N$ (see \cref{thm:tr-t1}).
This has important consequences for performing individual-based simulations that typically run the process several times and report an empirical average.
The limitation of naive individual-based simulations is that, a priori, it is not clear how much time will be needed until the simulations converge, and deciding to stop the simulations mid-way may bias the empirical average by over-representing the evolutionary trajectories that quickly go extinct.
Using \cref{thm:tr-t1}, this limitation can be circumvented by first efficiently computing an upper bound on the expected fixation time without having to simulate the process even once.

Third, we identify a class of population structures for which fixation times are \short{} for any $r\ge 1$. This class is surprisingly broad.
To start with, it includes several families of graph that had been studied in the context of Evolutionary Graph Theory earlier, such as Superstars and Metafunnels \cite{lieberman2005evolutionary}, or directed Fans \cite{allen2020transient}. 
Furthermore, the class also includes several other graph families of general interest, such as book graphs or cyclic complete multipartite graphs~\cite{Monk_van_Schaik_2022}.
Similarly, we prove that the fixation times on Megastars~\cite{galanis2017amplifiers} are \short{} for all $r\ge 1$ too.

While the focus of this work is to identify regimes and population structures that lead to \short{} fixation times, population structures with long fixation times may also be desirable, e.g.\ in conservation ecology to maintain high levels of ecological diversity~\cite{levene1953genetic,bulmer1972multiple,felsenstein1976theoretical,yeaman2011establishment,svoboda2023coexistence}.
Our results imply that spatial structures that support coexistence of two competing types on exponential time-scales all have a common feature. 
Namely, there must exist a starting node such that in the neutral regime ($r=1$), the fixation probability of a single mutant who initially appears at that node is exponentially small.
In other words, spatial structures in which a neutral mutant has a non-negligible chance of fixating, no matter where it appears, never support coexistence on long time-scales.

We note that throughout this work we considered the standard model of Moran process with Birth-death updating. A natural direction for future research is to consider related models, such as those with location-dependent fitness \cite{maciejewski2014environmental,kaveh2019environmental,svoboda2023coexistence} or those with death-Birth updating. 
It is known that in terms of fixation probabilities the Birth-death and the death-Birth processes behave quite differently \cite{sood2008voter,hadjichrysanthou2011evolutionary, hindersin2015most,tkadlec2020limits,brendborg2022fixation}.
However, in terms of fixation time they are qualitatively similar \cite{diaz2014approximating,durocher2022invasion}.

\section{Data and code availability}

Code for the figures and the computational experiments
is available from the Figshare repository: \url{https://doi.org/10.6084/m9.figshare.23802531}.

\section{Acknowledgements}
We thank Brendan McKay for helpful instructions on using \texttt{nauty}
and Salil Vadhan for insightful discussions about
random walks on directed graphs.
J.T. was supported by Center for Foundations of Modern Computer Science
(Charles Univ.\ project UNCE/SCI/004) and by the project PRIMUS/24/SCI/012 from Charles University.

\bibliographystyle{vancouver}
\bibliography{sources}

\begin{thebibliography}{10}

\bibitem{kimura1968evolutionary}
Kimura M.
\newblock Evolutionary rate at the molecular level.
\newblock Nature. 1968;217:624-6.

\bibitem{moran1958random}
Moran PAP.
\newblock Random processes in genetics.
\newblock In: Mathematical proceedings of the cambridge philosophical society.
  vol.~54. Cambridge University Press; 1958. p. 60-71.

\bibitem{ewens2004mathematical}
Ewens WJ.
\newblock Mathematical population genetics: theoretical introduction. vol.~27.
\newblock Springer; 2004.

\bibitem{durrett1994importance}
Durrett R, Levin S.
\newblock The importance of being discrete (and spatial).
\newblock Theoretical population biology. 1994;46(3):363-94.

\bibitem{whitlock2003fixation}
Whitlock MC.
\newblock Fixation probability and time in subdivided populations.
\newblock Genetics. 2003;164(2):767-79.

\bibitem{slatkin1981fixation}
Slatkin M.
\newblock Fixation probabilities and fixation times in a subdivided population.
\newblock Evolution. 1981:477-88.

\bibitem{Bonachela_Pringle_Sheffer_Coverdale_Guyton_Caylor_Levin_Tarnita_2015}
Bonachela JA, Pringle RM, Sheffer E, Coverdale TC, Guyton JA, Caylor KK, et~al.
\newblock Termite mounds can increase the robustness of dryland ecosystems to
  climatic change.
\newblock Science. 2015 Feb;347(6222):651–655.

\bibitem{Tilman_May_Lehman_Nowak_1994}
Tilman D, May RM, Lehman CL, Nowak MA.
\newblock Habitat destruction and the extinction debt.
\newblock Nature. 1994 Sep;371(64926492):65–66.

\bibitem{Comins_Hassell_May_1992}
Comins HN, Hassell MP, May RM.
\newblock The Spatial Dynamics of Host--Parasitoid Systems.
\newblock Journal of Animal Ecology. 1992;61(3):735–748.

\bibitem{Hassell_Comins_Mayt_1991}
Hassell MP, Comins HN, Mayt RM.
\newblock Spatial structure and chaos in insect population dynamics.
\newblock Nature. 1991 Sep;353(63416341):255–258.

\bibitem{Levin_1992}
Levin SA.
\newblock The Problem of Pattern and Scale in Ecology: The Robert H. MacArthur
  Award Lecture.
\newblock Ecology. 1992;73(6):1943–1967.

\bibitem{Levin_1976}
Levin SA.
\newblock Population Dynamic Models in Heterogeneous Environments.
\newblock Annual Review of Ecology and Systematics. 1976;7(1):287–310.

\bibitem{tilman1997spatial}
Tilman D, Kareiva P.
\newblock Spatial ecology: the role of space in population dynamics and
  interspecific interactions.
\newblock Princeton University Press; 1997.

\bibitem{Barton_1993}
Barton NH.
\newblock The probability of fixation of a favoured allele in a subdivided
  population.
\newblock Genetics Research. 1993 Oct;62(2):149–157.

\bibitem{Fisher_Ford_1950}
Fisher RA, Ford EB.
\newblock The "Sewall Wright effect.
\newblock Heredity. 1950;4:117–19.

\bibitem{Maruyama_1970}
Maruyama T.
\newblock Effective number of alleles in a subdivided population.
\newblock Theoretical Population Biology. 1970 Nov;1(3):273–306.

\bibitem{Nagylaki_Lucier_1980}
Nagylaki T, Lucier B.
\newblock NUMERICAL ANALYSIS OF RANDOM DRIFT IN A CLINE.
\newblock Genetics. 1980 Feb;94(2):497–517.

\bibitem{S_1932}
S W.
\newblock The roles of mutation, inbreeding, crossbreeding and selection in
  evolution.
\newblock Proceedings of the sixth international congress of Genetics.
  1932;1:356–366.

\bibitem{Wright_1931}
Wright S.
\newblock Evolution in Mendelian Populations.
\newblock Genetics. 1931 Mar;16(2):97–159.

\bibitem{Allen_Lippner_Chen_Fotouhi_Momeni_Yau_Nowak_2017}
Allen B, Lippner G, Chen YT, Fotouhi B, Momeni N, Yau ST, et~al.
\newblock Evolutionary dynamics on any population structure.
\newblock Nature. 2017 Apr;544(76497649):227–230.

\bibitem{Hauert_Doebeli_2004}
Hauert C, Doebeli M.
\newblock Spatial structure often inhibits the evolution of cooperation in the
  snowdrift game.
\newblock Nature. 2004 Apr;428(69836983):643–646.

\bibitem{Nakamaru_Matsuda_Iwasa_1997}
Nakamaru M, Matsuda H, Iwasa Y.
\newblock The Evolution of Cooperation in a Lattice-Structured Population.
\newblock Journal of Theoretical Biology. 1997 Jan;184(1):65–81.

\bibitem{Nowak_May_1992}
Nowak MA, May RM.
\newblock Evolutionary games and spatial chaos.
\newblock Nature. 1992 Oct;359(63986398):826–829.

\bibitem{Ohtsuki_Hauert_Lieberman_Nowak_2006}
Ohtsuki H, Hauert C, Lieberman E, Nowak MA.
\newblock A simple rule for the evolution of cooperation on graphs and social
  networks.
\newblock Nature. 2006 May;441(70927092):502–505.

\bibitem{Grenfell_Bjørnstad_Kappey_2001}
Grenfell BT, Bjørnstad ON, Kappey J.
\newblock Travelling waves and spatial hierarchies in measles epidemics.
\newblock Nature. 2001 Dec;414(68656865):716–723.

\bibitem{Lloyd_May_1996}
Lloyd AL, May RM.
\newblock Spatial Heterogeneity in Epidemic Models.
\newblock Journal of Theoretical Biology. 1996 Mar;179(1):1–11.

\bibitem{May_Lloyd_2001}
May RM, Lloyd AL.
\newblock Infection dynamics on scale-free networks.
\newblock Physical Review E. 2001 Nov;64(6):066112.

\bibitem{May_Nowak_1994}
May RM, Nowak MA.
\newblock Superinfection, Metapopulation Dynamics, and the Evolution of
  Diversity.
\newblock Journal of Theoretical Biology. 1994 Sep;170(1):95–114.

\bibitem{Noble_Burri_Le_Sueur_Lemant_Viossat_Kather_Beerenwinkel_2022}
Noble R, Burri D, Le~Sueur C, Lemant J, Viossat Y, Kather JN, et~al.
\newblock Spatial structure governs the mode of tumour evolution.
\newblock Nature Ecology \& Evolution. 2022 Feb;6(22):207–217.

\bibitem{Waclaw_Bozic_Pittman_Hruban_Vogelstein_Nowak_2015}
Waclaw B, Bozic I, Pittman ME, Hruban RH, Vogelstein B, Nowak MA.
\newblock A spatial model predicts that dispersal and cell turnover limit
  intratumour heterogeneity.
\newblock Nature. 2015 Sep;525(75687568):261–264.

\bibitem{Barton_1979}
Barton NH.
\newblock The dynamics of hybrid zones.
\newblock Heredity. 1979 Dec;43(33):341–359.

\bibitem{Mitchell_1974}
Mitchell JC.
\newblock Social Networks.
\newblock Annual Review of Anthropology. 1974;3(1):279–299.

\bibitem{lieberman2005evolutionary}
Lieberman E, Hauert C, Nowak MA.
\newblock Evolutionary dynamics on graphs.
\newblock Nature. 2005;433(7023):312-6.

\bibitem{diaz2021survey}
D{\'\i}az J, Mitsche D.
\newblock A survey of the modified Moran process and evolutionary graph theory.
\newblock Computer Science Review. 2021;39:100347.

\bibitem{donnelly1983finite}
Donnelly P, Welsh D.
\newblock Finite particle systems and infection models.
\newblock In: Mathematical Proceedings of the Cambridge Philosophical Society.
  vol.~94. Cambridge University Press; 1983. p. 167-82.

\bibitem{hindersin2015most}
Hindersin L, Traulsen A.
\newblock Most undirected random graphs are amplifiers of selection for
  birth-death dynamics, but suppressors of selection for death-birth dynamics.
\newblock PLoS computational biology. 2015;11(11):e1004437.

\bibitem{broom2010evolutionary}
Broom M, Hadjichrysanthou C, Rycht{\'a}{\v{r}} J.
\newblock Evolutionary games on graphs and the speed of the evolutionary
  process.
\newblock Proceedings of the Royal Society A: Mathematical, Physical and
  Engineering Sciences. 2010;466(2117):1327-46.

\bibitem{broom2010two}
Broom M, Hadjichrysanthou C, Rycht{\'a}{\v{r}} J, Stadler B.
\newblock Two results on evolutionary processes on general non-directed graphs.
\newblock Proceedings of the Royal Society A: Mathematical, Physical and
  Engineering Sciences. 2010;466(2121):2795-8.

\bibitem{monk2020wald}
Monk T, van Schaik A.
\newblock Wald’s martingale and the conditional distributions of absorption
  time in the Moran process.
\newblock Proceedings of the Royal Society A. 2020;476(2241):20200135.

\bibitem{monk2021martingales}
Monk T, van Schaik A.
\newblock Martingales and the characteristic functions of absorption time on
  bipartite graphs.
\newblock Royal Society Open Science. 2021;8(10):210657.

\bibitem{nowak2003linear}
Nowak MA, Michor F, Iwasa Y.
\newblock The linear process of somatic evolution.
\newblock Proceedings of the National Academy of Sciences.
  2003;100(25):14966-9.

\bibitem{nowak2006evolutionary}
Nowak MA.
\newblock Evolutionary dynamics: exploring the equations of life.
\newblock Harvard University Press; 2006.

\bibitem{adlam2015amplifiers}
Adlam B, Chatterjee K, Nowak MA.
\newblock Amplifiers of selection.
\newblock Proceedings of the Royal Society A: Mathematical, Physical and
  Engineering Sciences. 2015;471(2181):20150114.

\bibitem{hadjichrysanthou2011evolutionary}
Hadjichrysanthou C, Broom M, Rycht{\'a}r J.
\newblock Evolutionary games on star graphs under various updating rules.
\newblock Dynamic Games and Applications. 2011;1(3):386-407.

\bibitem{monk2014martingales}
Monk T, Green P, Paulin M.
\newblock Martingales and fixation probabilities of evolutionary graphs.
\newblock Proceedings of the Royal Society A: Mathematical, Physical and
  Engineering Sciences. 2014;470(2165):20130730.

\bibitem{chalub2014asymptotic}
Chalub FA.
\newblock Asymptotic expression for the fixation probability of a mutant in
  star graphs.
\newblock arXiv preprint arXiv:14043944. 2014.

\bibitem{pavlogiannis2017amplification}
Pavlogiannis A, Tkadlec J, Chatterjee K, Nowak MA.
\newblock Amplification on undirected population structures: comets beat stars.
\newblock Scientific reports. 2017;7(1):82.

\bibitem{diaz2013fixation}
D{\'\i}az J, Goldberg LA, Mertzios GB, Richerby D, Serna M, Spirakis PG.
\newblock On the fixation probability of superstars.
\newblock Proceedings of the Royal Society A: Mathematical, Physical and
  Engineering Sciences. 2013;469(2156):20130193.

\bibitem{jamieson2015fixation}
Jamieson-Lane A, Hauert C.
\newblock Fixation probabilities on superstars, revisited and revised.
\newblock Journal of Theoretical Biology. 2015;382:44-56.

\bibitem{galanis2017amplifiers}
Galanis A, G{\"o}bel A, Goldberg LA, Lapinskas J, Richerby D.
\newblock Amplifiers for the Moran process.
\newblock Journal of the ACM (JACM). 2017;64(1):1-90.

\bibitem{tkadlec2019population}
Tkadlec J, Pavlogiannis A, Chatterjee K, Nowak MA.
\newblock Population structure determines the tradeoff between fixation
  probability and fixation time.
\newblock Communications biology. 2019;2(1):138.

\bibitem{pavlogiannis2018construction}
Pavlogiannis A, Tkadlec J, Chatterjee K, Nowak MA.
\newblock Construction of arbitrarily strong amplifiers of natural selection
  using evolutionary graph theory.
\newblock Communications biology. 2018;1(1):71.

\bibitem{tkadlec2021fast}
Tkadlec J, Pavlogiannis A, Chatterjee K, Nowak MA.
\newblock Fast and strong amplifiers of natural selection.
\newblock Nature Communications. 2021;12(1):4009.

\bibitem{sharma2022suppressors}
Sharma N, Traulsen A.
\newblock Suppressors of fixation can increase average fitness beyond
  amplifiers of selection.
\newblock Proceedings of the National Academy of Sciences.
  2022;119(37):e2205424119.

\bibitem{cormen2022introduction}
Cormen TH, Leiserson CE, Rivest RL, Stein C.
\newblock Introduction to algorithms.
\newblock MIT press; 2022.

\bibitem{diaz2014approximating}
D{\'\i}az J, Goldberg LA, Mertzios GB, Richerby D, Serna M, Spirakis PG.
\newblock Approximating fixation probabilities in the generalized moran
  process.
\newblock Algorithmica. 2014;69:78-91.

\bibitem{goldberg2020phase}
Goldberg LA, Lapinskas J, Richerby D.
\newblock Phase transitions of the Moran process and algorithmic consequences.
\newblock Random Structures \& Algorithms. 2020;56(3):597-647.

\bibitem{diaz2016absorption}
D{\'\i}az J, Goldberg LA, Richerby D, Serna M.
\newblock Absorption time of the Moran process.
\newblock Random Structures \& Algorithms. 2016;49(1):137-59.

\bibitem{ibsen2015computational}
Ibsen-Jensen R, Chatterjee K, Nowak MA.
\newblock Computational complexity of ecological and evolutionary spatial
  dynamics.
\newblock Proceedings of the National Academy of Sciences.
  2015;112(51):15636-41.

\bibitem{Altrock_Gokhale_Traulsen_2010}
Altrock PM, Gokhale CS, Traulsen A.
\newblock Stochastic slowdown in evolutionary processes.
\newblock Physical Review E. 2010 Jul;82(1):011925.

\bibitem{maciejewski2014reproductive}
Maciejewski W.
\newblock Reproductive value in graph-structured populations.
\newblock Journal of Theoretical Biology. 2014;340:285-93.

\bibitem{monk2018martingales}
Monk T.
\newblock Martingales and the fixation probability of high-dimensional
  evolutionary graphs.
\newblock Journal of theoretical biology. 2018;451:10-8.

\bibitem{Monk_van_Schaik_2022}
Monk T, van Schaik A.
\newblock Martingales and the fixation time of evolutionary graphs with
  arbitrary dimensionality.
\newblock Royal Society Open Science. 2022 May;9(5):220011.

\bibitem{allen2020transient}
Allen B, Sample C, Jencks R, Withers J, Steinhagen P, Brizuela L, et~al.
\newblock Transient amplifiers of selection and reducers of fixation for
  death-Birth updating on graphs.
\newblock PLoS computational biology. 2020;16(1):e1007529.

\bibitem{mckey1990nauty}
McKay BD, Piperno A.
\newblock Practical graph isomorphism, II.
\newblock Journal of symbolic computation. 2014;60:94-112.

\bibitem{moller2019exploring}
M{\"o}ller M, Hindersin L, Traulsen A.
\newblock Exploring and mapping the universe of evolutionary graphs identifies
  structural properties affecting fixation probability and time.
\newblock Communications biology. 2019;2(1):137.

\bibitem{goldberg2019asymptotically}
Goldberg LA, Lapinskas J, Lengler J, Meier F, Panagiotou K, Pfister P.
\newblock Asymptotically optimal amplifiers for the Moran process.
\newblock Theoretical Computer Science. 2019;758:73-93.

\bibitem{downey2012parameterized}
Downey RG, Fellows MR.
\newblock Parameterized complexity.
\newblock Springer Science \& Business Media; 2012.

\bibitem{cygan2015parameterized}
Cygan M, Fomin FV, Kowalik {\L}, Lokshtanov D, Marx D, Pilipczuk M, et~al.
\newblock Parameterized algorithms. vol.~4.
\newblock Springer; 2015.

\bibitem{levene1953genetic}
Levene H.
\newblock Genetic equilibrium when more than one ecological niche is available.
\newblock The American Naturalist. 1953;87(836):331-3.

\bibitem{bulmer1972multiple}
Bulmer M.
\newblock Multiple niche polymorphism.
\newblock The American Naturalist. 1972;106(948):254-7.

\bibitem{felsenstein1976theoretical}
Felsenstein J.
\newblock The theoretical population genetics of variable selection and
  migration.
\newblock Annual review of genetics. 1976;10(1):253-80.

\bibitem{yeaman2011establishment}
Yeaman S, Otto SP.
\newblock Establishment and maintenance of adaptive genetic divergence under
  migration, selection, and drift.
\newblock Evolution. 2011;65(7):2123-9.

\bibitem{svoboda2023coexistence}
Svoboda J, Tkadlec J, Kaveh K, Chatterjee K.
\newblock Coexistence times in the Moran process with environmental
  heterogeneity.
\newblock Proceedings of the Royal Society A. 2023;479(2271):20220685.

\bibitem{maciejewski2014environmental}
Maciejewski W, Puleo GJ.
\newblock Environmental evolutionary graph theory.
\newblock Journal of theoretical biology. 2014;360:117-28.

\bibitem{kaveh2019environmental}
Kaveh K, McAvoy A, Nowak MA.
\newblock Environmental fitness heterogeneity in the Moran process.
\newblock Royal Society open science. 2019;6(1):181661.

\bibitem{sood2008voter}
Sood V, Antal T, Redner S.
\newblock Voter models on heterogeneous networks.
\newblock Physical Review E. 2008;77(4):041121.

\bibitem{tkadlec2020limits}
Tkadlec J, Pavlogiannis A, Chatterjee K, Nowak MA.
\newblock Limits on amplifiers of natural selection under death-Birth updating.
\newblock PLoS computational biology. 2020;16(1):e1007494.

\bibitem{brendborg2022fixation}
Brendborg J, Karras P, Pavlogiannis A, Rasmussen AU, Tkadlec J.
\newblock Fixation maximization in the positional moran process.
\newblock In: Proceedings of the AAAI Conference on Artificial Intelligence.
  vol.~36; 2022. p. 9304-12.

\bibitem{durocher2022invasion}
Durocher L, Karras P, Pavlogiannis A, Tkadlec J.
\newblock Invasion dynamics in the biased voter process.
\newblock In: Proceedings of the Thirty-First International Joint Conference on
  Artificial Intelligence; 2022. p. 265-71.

\bibitem{kotzing2019first}
K{\"o}tzing T, Krejca MS.
\newblock First-hitting times under drift.
\newblock Theoretical Computer Science. 2019;796:51-69.

\end{thebibliography}

\newpage

\appendix
\section*{Appendices}
This is a supplementary information to the manuscript \textit{Fixation times on directed graphs}. 
It contains formal proofs of the theorems listed in the main text.
In \cref{sec:prelims} we formally introduce the model and recall results that we will use in our proofs.
In the following sections we prove Theorems 1 to 5 from the main text.

\section{Preliminaries}\label{sec:prelims}
\subsection{Moran process on a graph}
A \emph{graph} $\Graph$ is a tuple $(\Vertices, \Edges)$
where $\Vertices$ is a set of \emph{vertices} and $\Edges$ is a set of \emph{edges} between vertices.
Unless otherwise specified, graphs are directed, unweighted, strongly connected, do not have self-loops,
and do not have multiple edges between vertices.
The \emph{outdegree} of a vertex $v$, denoted $\deg^+(v)$, is the number of outgoing edges, and
its \emph{indegree}, denoted $\deg^-(v)$, is the number of incoming edges. When the outdegree equals the indegree, we call it the \emph{degree} and denote it $\deg(v)$. 
Given a graph, its population size is $N := |\Vertices|$.
To emphasize the population size of a graph $\Graph$, we sometimes
refer to it as $\Graph_N$.
To model the evolutionary dynamics we consider the Moran Birth-death processes on graphs.
In this process, each vertex has an associated \emph{type} at every time step: it is either a \emph{resident} or a \emph{mutant}.
Each type has an associated reproductive fitness: residents have fitness 1, mutants have fitness $r\ge 1$.

The \emph{Moran Birth-death process} is defined as follows: At each time step we
\begin{enumerate}
    \item Pick a random vertex $u\in\Vertices$ proportional to its type's reproductive fitness.
    \item Pick an outgoing edge of $u$ uniformly at random; denote its endpoint $v\in\Vertices$.
    \item Update the type of $v$ to be the type of $u$.
\end{enumerate}
Formally, the process is represented as a sequence $X_0,X_1,X_2,\ldots \subseteq \Vertices$, where $X_0$ is the set of vertices that are initially occupied by mutants, and $X_t$ is the set of vertices occupied by mutants after $t$ steps. The set of vertices occupied by mutants is called a \textit{(mutant) configuration}.
If $\Graph$ is finite and strongly connected, then with probability one eventually all vertices are occupied by mutants (if $X_t=\Vertices$ for some $t$, we say that the mutants \textit{fixated}) or all vertices are occupied by residents (if $X_t=\emptyset$ for some $t$, we say that the mutants \textit{went extinct}).

\subsection{Fixation probability and time}
We define some relevant quantities when studying evolutionary dynamics in a structured population.
Each quantity depends 
on the mutant fitness advantage $r\ge 1$; 
on the underlying population structure, represented by a graph $\Graph_N$; and
on the set $X_0$ of vertices that are occupied by mutants.

The \emph{fixation probability}, denoted $\fp:=\fp_r(\Graph,S)$, is the probability that the mutants eventually take over the population forever
starting from a set of mutants occupying vertices $S\subseteq\Vertices$.
In contrast,
the \emph{extinction probability} is the probability that the mutants eventually die out forever.
If the graph is strongly connected, fixation or extinction will happen with probability 1 in finitely many expected steps.
When the processes reaches fixation or extinction, we say that the process has \emph{absorbed}.
We can also study the time it takes for each of these events to happen:
the \emph{expected absorption time} $\ATr(\Graph,S)$ is the expected amount of time steps until the process absorbs.
The \emph{expected fixation time} $\FTr(\Graph,S)$ is the expected amount of time steps
conditional on fixation occurring.
We also define $\FTr(\Graph)=\max_{S\subseteq\Vertices, S\ne\emptyset} \{\FTr(\Graph,S)\}$ as the slowest possible fixation time, across all possible initial conditions.
Similarly, we define the \textit{expected extinction time}
$\ETr(G,S)$
as the average number of steps over those trajectories that eventually lead to mutant extinction.
When the mutant fitness advantage $r$, the underlying graph $\Graph$, or the initial condition $S$ are clear from the context we omit it.
Further, for a $u\in\Vertices$, we sometimes denote $\fp_r(\Graph,\{u\})$
as $\fp_r(\Graph,u)$ for notational ease.

We note that there are two types of steps in the Moran process, namely the \textit{active steps} in which the configuration changes, and the \textit{waiting steps} in which it stays the same. The absorption, fixation, and extinction times defined above count both the active steps and the waiting steps.

Finally, we say that if the expected absorption or fixation time is bounded from above by some polynomial in terms of the population size $N$ then the time is \emph{fast}; otherwise the time is \emph{slow}.

\subsection{Forward bias lemma}
In the rest of this section, we define the notion of a forward bias and present a lemma that proves useful in deriving~\cref{thm:r-infty,thm:eulerian}.

Given a strongly connected directed graph $\Graph_N=(\Vertices,\Edges)$, a mutant fitness advantage $r\geq 1$, and a nonempty set $S\subsetneq \Vertices$ of nodes currently occupied by mutants, we define the \textit{up-probability} $\pup(S)$ as the probability that, in a single step of the Moran process, the number of mutants increases.
In other words, $\pup(S)$ is the probability that a mutant is selected for reproduction and its offspring replaces a resident neighbor.
Likewise, we define the \textit{down-probability} $\pdown(S)$ as the probability that the number of mutants decreases.
Similarly to above, $\pdown(S)$ is the probability that a resident is selected for reproduction and its offspring replaces a mutant neighbor.
Since the graph $\Graph_N$ is strongly connected and $S\not\in\{\emptyset,\Vertices\}$ then both $\pup(S)$ and $\pdown(S)$ are non-zero.
In that case, we define the \textit{bias at $S$} as the ratio $\bias(S)=\pup(S)/\pdown(S)$.

\begin{lemma}\label{lem:key}
Let $\Graph_N=(\Vertices,\Edges)$ be a strongly connected directed graph with $N$ nodes, 
$S\subsetneq \Vertices$ the nonempty set of nodes initially occupied by mutants,
 and $r\ge 1$ the mutant fitness advantage.
Suppose that there exists a real number $f>1$ such that for any nonempty subset $U\subsetneq \Vertices$ of nodes we have $\bias(U)\ge f$. Then
\begin{enumerate}
\item\label{item:key-1} $\fpr(\Graph_N,S)\ge 1-1/f$,
\item\label{item:key-2} $\ATr(\Graph_N,S)\le \frac{f+1}{f-1}\cdot N^3$,
\item\label{item:key-3}$\FTr(\Graph_N,S)\le \frac{\ATr(\Graph_N,S)}{\fpr(\Graph_N,S)}\le \frac{f(f+1)}{(f-1)^2}\cdot N^3$.
\end{enumerate}
Further, for any $\e\in(0,1)$ and for parameters $r$, $\Graph_N$,
 and $S$,
the amount of time until absorption given the process fixates is at most $\frac{f(f+1)}{(f-1)^2}\cdot \frac{N^3}\e$ with probability at least $1-\e$.
\end{lemma}

To prove~\cref{lem:key} we need the following form of Markov's inequality.
\begin{lemma}[Conditional Markov's inequality]\label{lem:markov}
    Let $X$ be an almost surely nonnegative random variable.
    Let $a>0$ and let $\mathcal E$ be an event.
    Then
   \begin{equation}
       \Pr[X>a\mid \mathcal E]\leq \E{X/a\mid \mathcal E}.
   \end{equation} 
\end{lemma}
\begin{proof}
    Similar to a common proof of Markov's inequality, notice $a\cdot\mathbbm1_{\{X> a\}}\cdot\mathbbm1_{\{\mathcal E\}}\leq X\cdot\mathbbm1_{\{\mathcal E\}}$. Taking the expectation of both sides
    yields the result.
\end{proof}

We proceed to prove~\cref{lem:key}.
\begin{proof}[Proof of \cref{lem:key}]
The key idea is to project the process to a one-dimensional random walk by tracking the number of nodes occupied by mutants. Formally, the random walk is a Markov chain $W$ with states $s_0,s_1,\dots,s_N$ (where state $s_k$ corresponds to those mutant configurations with precisely $k$ mutants), and with transition probabilities $\Pr[s_k\to s_{k+1}]=\frac{f_k}{f_k+1}$, $\Pr[s_k\to s_{k-1}]=\frac{1}{f_k+1}$, where $f_k=\min\{\bias(S)\colon |S|=k\}$ is the smallest forward bias among all the mutant configurations with precisely $k$ mutants.
For $k=0,\ldots,N$, let $p_k$ be the probability that $W$ starting at $s_k$ eventually reaches $s_N$ (as opposed to reaching $s_0$).
Note that the random walk $W$ models only the active steps of the Moran process (that is, the steps in which the mutant configuration changes).
Moreover, at each configuration it always assumes the lowest possible forward bias.
Let $i=|S|$ be the number of nodes initially occupied by mutants.
Thus, we have $\fpr(\Graph_N,S)\ge p_i$.

\medskip
First, we prove~\cref{item:key-1}.
A standard formula for one-dimensional Markov chains (see e.g.~\cite[Section 6.2]{nowak2006evolutionary} 
immediately yields the desired
\begin{align*}
p_i \geq p_1 = \frac1{1+\sum_{j=1}^{N-1}\prod_{k=1}^j \frac1{f_k}} \ge \frac1{\sum_{j=0}^{N-1} (1/f)^j}\ge \frac1{\sum_{j=0}^\infty (1/f)^j} = 1-1/f.
\end{align*}

\medskip
Next, we prove~\cref{item:key-2}.
We deal with the active steps and the waiting steps separately.

Regarding the waiting steps, consider any nonempty mutant configuration $S\subsetneq \Vertices$ with $1\le k\le N-1$ mutants and let $F=k\cdot r+(N-k)\cdot 1\le r\cdot N$ be the total fitness of the population. 
Since $\Graph_N$ is strongly connected, there is at least one edge going from a mutant node to a resident node.
Since $r\ge 1$, the probability $p_a(S)$ that the next step is active satisfies $p_a(S)\ge \frac{r}{F}\cdot \frac1{N-1}\ge \frac1N\cdot \frac1N =1/N^2$.
Thus, at any configuration, the expected number of steps until an active step occurs is at most $N^2$.
Therefore, in order to get an upper bound on the absorption time (which includes both the active and the waiting steps), it suffices to count only the active steps and then multiply the result by $N^2$.

To count the active steps, consider the corresponding random walk $W$. Given $1\le k\le N-1$, let $x_k$ be the expected number of times the state $s_k$ is visited in $W$. We will prove that $x_k\le \frac{f+1}{f-1}$ for each $1\le k\le N-1$.

To that end, consider a walk $W$ currently at $s_k$.
With probability $\frac{f}{f+1}$ it next moves to $s_{k+1}$. 
Once in $s_{k+1}$, by~\cref{item:key-1} the walk reaches $s_N$ before reaching $s_k$ with probability at least $p_1\ge 1-1/f$.
Thus, any time the walk is at $s_k$, with probability at least $\frac{f}{f+1}\cdot \frac{f-1}{f}=\frac{f-1}{f+1}$ it never comes back. Therefore, $x_k\le 1/\frac{f-1}{f+1} =\frac{f+1}{f-1}$. This is true for any of the $N-1$ states $s_1,\dots,s_{N-1}$, so the expected number of active steps is at most $\frac{f+1}{f-1}\cdot N$ and the expected number of all steps (including the waiting steps) is at most
$\ATr(\Graph_N,S)\le N^2\cdot\left(\frac{f+1}{f-1}\cdot N\right) = \frac{f+1}{f-1}\cdot N^3$.

\medskip

Finally, we prove~\cref{item:key-3}. By linearity of expectation we have
\[
\ATr(\Graph_N,S) = \FTr(\Graph_N,S)\cdot \fpr(\Graph_N,S) + \ETr(\Graph_N,S)\cdot \left(1-\fpr(\Graph_N,S)\right),
\]
where $\ETr(\Graph_N,S)$ is the conditional extinction time, that is, the average length of those stochastic trajectories that terminate with the mutation going extinct. Applying a trivial bound $\ETr(\Graph_N,S)\ge 0$ and~\cref{item:key-1}, we get
\begin{equation}\label{eq:fixation-through-absorption}
\ATr(\Graph_N,S) \ge \FTr(\Graph_N,S)\cdot \fpr(\Graph_N,S) \ge \FTr(G_n,S)\cdot (1-1/f).
\end{equation}
Putting this together with~\cref{item:key-2} gives the desired
\[
\FTr(G_n,S)\le \frac{f}{f-1}\cdot \ATr(\Graph_N,S)\le \frac{f(f+1)}{(f-1)^2}\cdot N^3.
\]

 For the last statement of the lemma, we apply \cref{lem:markov}
to the time a trajectory takes to absorb as $X$, the event the process fixates as $\mathcal E$,
and $a = \e^{-1}\cdot \E{X\mid \mathcal E}$ using \cref{item:key-3}.

\end{proof}

%




\section{Fixation always occurs quickly when selection advantage is strong enough}\label{sec:r-infty}

In this section we prove Theorem 1 from the main text.
That is, we show that if the mutant fitness advantage is large enough then the process terminates fast, regardless of the underlying spatial structure.

\begin{theorem}\label{thm:r-infty} Let $\Graph_N$ be an arbitrary graph on $N$ nodes.
Suppose that $r\ge N^2$.
Then $\ATr(\Graph_N)\le 2\cdot N^3$ and  $\FTr(\Graph_N)\le 3\cdot N^3$.
\end{theorem}
\begin{proof} Suppose mutants currently occupy a nonempty set $S\subsetneq \Vertices$ of nodes.
Let $k=|S|$ and denote by $F=k\cdot r+ (N-k)\cdot 1$ the total fitness of the population.

Since $\Graph_N$ is strongly connected, there is at least one mutant node $u$ with a resident neighbor.
We claim that if $u$ is selected for reproduction, then it replaces a resident with probability at least $1/k$.
We distinguish two cases based on the outdegree $d$ of $u$.
\begin{enumerate}
\item If $d\le k$ then the claimed probability is at least $\frac1d\ge \frac1k$.
\item If $d\ge k$ then $u$ must have at least $d-(k-1)$ resident neighbors (since there are at most $k-1$ other mutant nodes altogether). Thus the claimed probability is at least $\frac{d-k+1}{d}\ge \frac1k$, where the inequality is equivalent with $(d-k)(k-1)\ge 0$ which holds trivially.
\end{enumerate}
Altogether, node $u$ gets selected with probability at least $r/F$, thus we have:
\[\pup(S)\ge \frac{r}{F}\cdot \frac1k.
\]

On the other hand, since there are $N-k$ residents and each, when selected for reproduction, replaces a mutant with probability at most 1, we have
\[\pdown(S)\le \frac{N-k}{F}.
\]
Combining those two bounds, we get
\[ \bias(S)\ge \frac{r}{k(N-k)} \ge \frac{r}{N^2/4}\ge 4,
\]
where we have used an AM-GM inequality for $k$ and $N-k$. Thus, \cref{lem:key} applies with $f=4$ and we get the desired  $\ATr(\Graph_N)\le \frac53 N^3<2N^3$ and  $\FTr(\Graph_N)\le \frac{20}{9}N^3<3N^3$.
\end{proof}

\section{Fixation occurs quickly on Eulerian graphs}\label{sec:eulerian}
In this section we prove Theorem 2 from the main text.
That is, we show that if a graph is Eulerian with degrees sandwiched between $\delta$ and $\Delta$ then the time is \short{}, provided that the mutant fitness advantage satisfies $r> \Delta/\delta$.
Recall that a graph is \emph{Eulerian} if $\InDegree{v}=\OutDegree{v}$ holds for each $v\in\Vertices$.

First, we point out one useful property of such graphs.
Let $S\subseteq \Vertices$.
Let $E^+(S)$ be the set of the edges whose starting vertex is in $S$.
Let $E^-(S)$ be the set of the edges whose ending vertex is in $S$.
Let $m^+_S$ be the number edges outgoing from $S$
and incoming to $\Vertices\setminus S$.
Similarly, let $m^-_S$ be the number edges outgoing from $\Vertices\setminus S$
and incoming to $S$.

\begin{lemma}\label{lemma:eulerian-cross-edges}
A graph $\Graph=(\Vertices,\Edges)$ is Eulerian
if and only if $m^+_S = m^-_S$ for every $S\subseteq\Vertices$.
\end{lemma}
\begin{proof}
    Suppose $\Graph$ is Eulerian and let $S\subseteq\Vertices$.
    We have that
    
    \begin{align*}
        m^+_S &= |\Edges^+(S)|-\#\{u\to v\in\Edges\mid u\in S, v\in S\} \text{ \quad and} \\
       m^-_S &= |\Edges^-(S)|-\#\{u\to v\in\Edges\mid u\in S, v\in S\}.
    \end{align*}
    Noticing that \[ |\Edges^+(S)|=\sum_{u\in S}\OutDegree{u}=\sum_{u\in S}\InDegree{u}=|\Edges^-(S)| \]
    since $\Graph$ is Eulerian completes one direction of the proof.

    For the other direction, suppose we know that $m^+(S)=m^-(S)$ for each $S\subseteq\Vertices$.
    Then for every $u\in\Vertices$ we know $\InDegree{u} = m^-_{\{u\}} =m^+_{\{u\}} =\OutDegree{u}$, so $\Graph$ is Eulerian.
\end{proof}


\begin{theorem}\label{thm:eulerian}
    Let $\Graph_N$ be an Eulerian graph on $N$ nodes
    with smallest degree $\delta$ and largest degree $\Delta$.
    Suppose that $r\ge \frac{\Delta}{\delta}\cdot (1+\eps)$ for some $\eps>0$.
    Then
    $\ATr(\Graph_N)\le \frac{2+\eps}{\eps}\cdot N^3$ and $\FTr(\Graph_N)\le \frac{(1+\eps)(2+\eps)}{\eps^2}\cdot N^3$. 
\end{theorem}
\begin{proof}
Suppose mutants currently occupy a nonempty set $S\subsetneq \Vertices$ of nodes.
Let $E^+=\{ (u,v)\in \Edges\mid u\in S, v\not\in S\}|$ be the set of those edges that go from a mutant to a resident.
Likewise, let $E^-=\{ (u,v)\in \Edges\mid u\not\in S, v\in S\}$ be the set of those edges that go from a resident to a mutant.
By \cref{lemma:eulerian-cross-edges} we know that the two sets $E^+$ and $E^-$ have the same size, denote it by $s$.
Denoting by $F=|S|\cdot r + (N-|S|)\cdot 1$ the total fitness of the population, we have
\[
\pup(S)\ge \sum_{(u,v)\in E^+} \frac{r}{F}\cdot \frac{1}{\deg(u)} \ge s\cdot \frac{r}{F\Delta}
\]
and
\[
\pdown(S)\le \sum_{(u,v)\in E^-} \frac{1}{F}\cdot \frac{1}{\deg(u)} \le s\cdot \frac{1}{F\delta}.
\]
Thus
\[\bias(S)\ge \frac{r\delta}{\Delta} = 1+\eps,\]
hence \cref{lem:key} applies with $f=1+\eps$.
The claims follow by straightforward algebra.
\end{proof}

\section{Fixation can occur slightly faster for small selective advantages}\label{sec:tr-t1}

In this section we prove Theorem 3 from the main text.
That is, we bound the fixation time on any graph $\Graph_N$ with any mutant fitness advantage $r\ge 1$ in terms of the minimum fixation probability $\pmin := \min_{S\subseteq\Vertices,S\neq\emptyset}\fp_{r=1}(\Graph,S)$.
The proof has two ingredients.
As our first ingredient, we bound the fixation time in the neutral regime ($r=1$). 
To that end, we first recall a standard lemma.

\begin{lemma}[Corollary 26 of \cite{kotzing2019first}, martingale upper additive drift]\label{lemma:first-hitting-times}
Let $Z_0,Z_1,Z_2,\ldots$ be random variables over
$[\alpha, \beta] \subseteq \mathbb R$,
and let $\tau := \inf \{t \geq 0 \mid Z_t \in \{\alpha, \beta\}\}$.
Furthermore, suppose that,
\begin{enumerate}
    \item $\E{Z_{t+1} \mid Z_0,\ldots, Z_t} = Z_t$ for all $t < \tau$, and
    \item there is some $\delta>0$ such that
    $\Var{Z_{t+1} | Z_0,\ldots, Z_t} \geq \delta$
    for all $t < \tau$.
\end{enumerate}
Then
\begin{equation}
    \E{\tau\mid Z_0} \leq \frac{(Z_0-\alpha)(\beta-Z_0)}{\delta}.
\end{equation}
\end{lemma}


\cref{lemma:first-hitting-times} allows us to bound the absorption time and the fixation time in the neutral regime ($r=1$). 

\begin{lemma}[Bounding absorption time and fixation time when $r=1$]\label{lemma:fixation-time-through-absorption-time}
    Let $\Graph=(\Vertices,\Edges)$. Then
    \[
    \AbsorptionTime_{r=1}(\Graph)\leq \frac{N^4}{4\pmin^2}
    \quad\text{and}\quad
    \T_{r=1}(\Graph) \leq \frac{N^4}{4\pmin^3}.
    \]
\end{lemma}
\begin{proof}
    Note that for any nonempty set $S\subsetneq \Vertices$ of mutants, the probability of making an active step is at least $1/N^2$. Indeed, since $\Graph$ is strongly connected, there is a mutant node $u$ with a resident out-neighbor $v$. Node $u$ is selected for reproduction with probability $r/(r\cdot |S|+ (N-|S|))\ge 1/N$, and the offspring replaces $v$ with probability $1/\deg^+(u)>1/N$. Moreover, since in the neutral regime ($r=1$) the fixation probability is additive over the set of nodes occupied by mutants, if such an active step happens and node $v$ becomes a mutant, then the fixation probability increases by $\fp_{r=1}(\Graph,\{v\})\ge \fpmin$.

    Let $(X_t)_{t\ge 0}$ be the mutant configuration after $t$ steps of the Moran process. 
    We aim to apply~\cref{lemma:first-hitting-times}, where $Z_t=\fp(X_t)$, $\alpha=0$, and $\beta=1$.
    To do that, we bound the conditional variance $\Var{Z_{t+1}\mid Z_0,\dots,Z_t}$ from below as follows:
     \begin{align}
         \Var{Z_{t+1}\mid Z_0,\ldots,Z_t}
         &= \E{\left(Z_{t+1}-\E{Z_{t+1}\mid Z_0,\ldots,Z_t}\right)^2\mid Z_0,\ldots,Z_t} \\
         &= \E{\left(Z_{t+1}-Z_{t}\right)^2\mid Z_0,\ldots,Z_t} \\
         &\geq \left(\E{\big|Z_{t+1}-Z_{t}\big|\mid Z_0,\ldots,Z_t}\right)^2 \\
         &\geq \left(\frac1{N^2}\cdot \fpmin\right)^2,
    \end{align}
    where in the respective steps we used the definition of conditional variance, the fact that fixation probability does not change in expectation in one step of the process, Jensen's inequality for a convex function $f(x)=x^2$, and the observation that with probability at least $1/N^2$ the fixation probability changes by at least $\fpmin$.

    Applying~\cref{lemma:first-hitting-times}, we find that 
    \begin{equation}\label{eq:absorption-time-upper-bound}
        \AbsorptionTime_{r=1}(\Graph,S) \leq \left(\frac{N^2}{\pmin}\right)^2\cdot\FixationProbability(S)(1-\FixationProbability(S)) \leq \frac{N^4}{4\pmin^2},
    \end{equation}
    for any $S\subseteq \Vertices$ where in the last step we used an inequality $x(1-x)\le 1/4$ that holds for any $x\in(0,1)$.

    
    

    Finally, we turn the obtained bound on absorption time into a bound on fixation time. Let $S\subseteq \Vertices$ and denote by $\ET_{r=1}(G,S)$ the extinction time starting from $S$. Then
    \[
    \AT_{r=1}(\Graph,S) = 
    \FixationProbability_{r=1}(\Graph, S)\cdot \T_{r=1}(\Graph,S)
    +\left(1-\FixationProbability_{r=1}(\Graph, S)\right)\cdot \ET_{r=1}(\Graph,S) \geq \pmin\cdot \T_{r=1}(\Graph,S),
    \]
    thus taking the maximum over $S\subseteq\Vertices$ we obtain $\T_{r=1}(\Graph)\leq \AT_{r=1}(\Graph)/\pmin$.
\end{proof}

As our second ingredient, we relate the fixation time in the neutral regime ($r=1$) and the fixation time when $r>1$.

\begin{lemma}\label{corollary:fixation-time-for-r-greater-than-1}
    Let $\Graph=(\Vertices,\Edges)$, $r\ge 1$, and let $\pmin := \min_{S\subseteq\Vertices, S\neq\emptyset}\fp_{r=1}(\Graph,S)$.
    Then $\FTr(\Graph) \leq \frac{ 4r}{\pmin}\cdot \FT_{r=1}(\Graph)$.
\end{lemma}
\begin{proof}
    Consider the process $M_1$ with $r=1$ and the process $M_r$ with $r\ge 1$.
    In $M_1$, there exist fixating trajectories with probability mass at least $\pmin$.
    They fixate on average in at most $T_1=\FT_{r=1}(\Graph)$ steps.
    Let $A_1$ be a random variable for the
    absorption time of a trajectory drawn from $M_1$ according
    to the birth-death updating rule.
    Let $\mathcal F$ be the event that a trajectory drawn from $M_1$ according
    to the birth-death updating rule fixates.
    Thus $T_1 = \E{A_1\mid \mathcal F}$.
    By applying \cref{lem:markov} with $X=A_1$, $a=2T_1$, and $\mathcal E=\mathcal F$,
    \begin{equation}
        \Pr[A_1> 2T_1\mid \mathcal F] \leq \E{A_1\mid \mathcal F}/(2T_1)
        = 1/2.
    \end{equation}
    So for the complementary event $A_1\le 2T_1$, we have
    \begin{align}
        \Pr[A_1 \leq 2T_1 \text{ and } \mathcal F]
        &= \Pr[A_1 \leq 2T_1 \mid \mathcal F] \cdot \Pr[\mathcal F] \\
        &\geq (1/2) \cdot \pmin.
    \end{align}
    Thus in $M_1$ there exist fixating trajectories with total probability mass at least $\pmin/2$ that fixate in at most $2 T_1$ steps each.
    
    Let $\tilde{M}_1$ and $\tilde{M}_r$ denote the continuous-time versions of the processes
    as described in \cite{diaz2016absorption}.
    Then by Lemma 5 of \cite{diaz2016absorption}, there is a coupling 
    between the continuous-time versions of the two processes such that
    if the processes start with the same set of mutants (that is, $\tilde{M}_1[0] = \tilde{M}_r[0]$),
    then $\tilde{M}_1[t] \subseteq \tilde{M}_r[t]$ for all $t \geq 0$.
    Let $\tilde{\tau}_1$ be one possible mutant set trajectory that fixates for $\tilde{M}_1$.
    Then the coupling implies that the corresponding trajectory $\tilde{\tau}_r$ for $\tilde{M}_r$
    fixates even earlier, when measured in the continuous time.
    
    Next, we transfer this relationship back into the world of discrete-time processes $M_1$ and $M_r$.
    Note that at each moment in time, the total fitness of the population in $\tilde{M}_r$ is at most $rN$, that is, it is at most $r$ times as large as the total fitness of the population in $\tilde{M}_1$.
    Thus, reproduction events in $\tilde{M}_r$ occur at a rate that is at most $r$ times larger than the rate at which reproduction events occur in $\tilde{M}_1$.
    When we move from continuous time to discrete time, we count each reproduction event as lasting 1 unit of time.
    Thus, any time a trajectory $\tilde{\tau}_1\in \tilde{M}_1$ gives rise to a trajectory $\tau_1\in M_1$ with length $\ell$, the coupled trajectory $\tilde{\tau}_r \in \tilde{M}_r$ gives rise to a trajectory $\tau_r\in M_r$ with length at most $r\ell$.
    
    Because in $M_1$ there exist fixating trajectories with total probability mass at least $\fpmin/2$ that fixate in at most $2T_1$ steps each,
    in $M_r$ there exist fixating trajectories with total probability mass $\pmin/2$ that fixate in at most $2rT_1$ steps each.
    Now imagine we run $M_r$ for stages of $2rT_1$ steps each.
    Within each stage, the process fixates with probability at least $\pmin/2$ so in expectation we observe fixation after at most $2/\pmin$ stages. 
    In total, this gives the desired
    \[
    \FTr(\Graph) \le \frac2{\pmin}\cdot 2r\cdot T_1 = \frac{4r}{\pmin}\cdot T_1(\Graph).\qedhere
    \]
\end{proof}

It remains to combine the two ingredients.

\begin{theorem}\label{thm:fixation-time-bound-using-only-rho-min}
Fix a graph $\Graph_N$ on $N$ vertices and $r\ge 1$.
Then $\FTr(\Graph_N)\le
\frac{N^6}{\pmin^4}.$
\end{theorem}
\begin{proof}
We distinguish two cases. If $r\ge N^2$ then \cref{thm:r-infty} implies that $\FTr(\Graph_N)\le 3N^3$ which is stronger than the claimed bound. So suppose $r< N^2$. Then \cref{corollary:fixation-time-for-r-greater-than-1,lemma:fixation-time-through-absorption-time} yield
\[\FTr(\Graph_N)
\le \frac{4r}{\fpmin}\cdot \FT_{r=1}(\Graph_N) \le \frac{4r}{\fpmin}\cdot \frac{N^4}{4\fpmin^3}
\le \frac{N^6}{\fpmin^4}. \qedhere
\]
\end{proof}


\section{Balanced graphs (and others) fixate quickly}\label{sec:good}
\begin{lemma}\label{lemma:fp-linear-system}
    Let $\Graph=(\Vertices,\Edges)$.
    Suppose there exist $|\Vertices|$ numbers $\{x_v\mid v\in \Vertices\}$ that satisfy
    \begin{equation}\label{eq:lin-system}
       x_v\cdot \sum_{u: u\to v\in \Edges} \frac1{\OutDegree{u}}
       =
       \frac1{\OutDegree{v}} \cdot \sum_{w: v\to w \in \Edges} x_w
    \end{equation}
    for each $v\in\Vertices$.
    In addition, suppose $\sum_{v\in\Vertices} x_v=1$.
    Then $x_v=\FixationProbability(\{v\})$.
    
\end{lemma}
\begin{proof}
Suppose $r=1$ and denote $N=|\Vertices|$. Then the fixation probability is additive, that is, $\fp(S) = \sum_{v\in S} \fp(\{v\})$ for every $S\subseteq \Vertices$ \cite{broom2010two}.
Thus, the list of $2^N$ fixation probabilities $\FixationProbability(S)$ for $S\subseteq \Vertices$ is determined by the list of $N$ fixation probabilities $\fp(\{v\})$ for $v\in\Vertices$.
The fixation probabilities $\fp(\{v\})$ are the unique solutions to the linear system
\begin{equation}
    \FixationProbability(\{v\}) = \frac1N\sum_{v\to w\in\Edges} \frac1{\OutDegree{v}}\cdot {\FixationProbability(\{v,w\})}
    +\frac1N\sum_{u\to v\in\Edges}\frac1{\OutDegree{u}}\cdot \FixationProbability(\emptyset)
    +\left(1-\frac1N - \frac1N\sum_{u\to v\in\Edges}\frac1{\OutDegree{u}}\right)\cdot \FixationProbability(\{v\}),
\end{equation}
where the first term on the right-hand side corresponds to the mutant reproducing, the second term corresponds to the mutant being replaced, and the third term corresponds to a resident replacing a resident.
Using $\fp(\emptyset)=0$ and $\fp(\{v,w\})=\fp(\{v\})+\fp(\{w\})$ we obtain 
    \begin{equation}\label{eq:lin-system}
       \FixationProbability(\Set{v})\cdot \sum_{u: u\to v\in \Edges} \frac1{\OutDegree{u}}
       =
       \frac1{\OutDegree{v}} \cdot \sum_{w: v\to w \in \Edges} \FixationProbability(\Set{w})
    \end{equation}
which is precisely the system satisfied by $\{x_v\mid v\in\Vertices\}$.

\end{proof}

\begin{definition}
   A graph $\Graph=(\Vertices,\Edges)$ is \emph{balanced} if and only if
   \begin{equation}
       \frac1{\InDegree{v}}\cdot \sum_{u: u\to v\in \Edges} \frac1{\OutDegree{u}}
       =
       \frac1{\OutDegree{v}} \cdot \sum_{w: v\to w \in \Edges} \frac1{\InDegree{w}}.
    \end{equation}
\end{definition}


\begin{theorem}\label{thm:balanced-graph-properties} 
    Let $\Graph_N$ be a \good{} graph. Then:
    \begin{enumerate}
    \item\label{item:balanced-graph-properties-1} $\fp_{r=1}(\Graph_N,u)=\frac{1/\indeg(u)}{\sum_{v\in \Vertices} 1/\indeg(v)}\ge 1/N^2$ for any node $u$.
    \item\label{item:balanced-graph-properties-2} $\FTr(\Graph_N)\le N^{14}$ for any $r\ge 1$.
    \end{enumerate}
\end{theorem}
\begin{proof}
    The equality in Item \ref{item:balanced-graph-properties-1} follows from Lemma \ref{lemma:fp-linear-system} and the bound follows from the fact that the numerator is at least $1/N$ and the denominator is at most $N$.
    Item \ref{item:balanced-graph-properties-2} follows immediately from Item \ref{item:balanced-graph-properties-1}
    and from \cref{thm:fixation-time-bound-using-only-rho-min}.
\end{proof}

\noindent In the rest of this section we verify that the undirected graphs, carousels, books, metafunnels, and superstars are all balanced. Thus, the fixation probability under neutral drift ($r=1$) starting from node $v$ is inversely proportional to $\indeg(v)$.
Moreover, we provide an explicit formula for fixation probability on Megastars under neutral drift.
This gives an upper bound on the fixation time for any $r\ge 1$ by \cref{thm:fixation-time-bound-using-only-rho-min}.
We use the notation $\FixationProbabilityUnderNeutralEvolution{\Set{v}}\propto p_v$ to mean that $\FixationProbabilityUnderNeutralEvolution{\Set{v}} = p_v/\sum_{w\in\Vertices} p_w$.;

Recall that a graph is \textit{undirected} if all edges are two-way, that is, if $u\to v\in\Edges$ then $v\to u\in\Edges$ as well.

\begin{claim}\label{thm:bi-directional-fps}
    Suppose $\Graph=(\Vertices, \Edges)$ is undirected.
    Then $\FixationProbabilityUnderNeutralEvolution{\Set{v}} \propto 1/\Degree(v)$.
\end{claim}
\begin{proof}
    Checking by substituting,
    \begin{align}
        &\frac{1}{\Degree(v)}\sum_{v\to w\in\Edges}\frac{1}{\Degree(w)}
        = \frac{1}{\Degree(v)} \sum_{u\to v\in\Edges}\frac{1}{\Degree(u)} \\
        \iff 
        &\frac{1}{\Degree(v)}\sum_{\substack{v\to w\in\Edges \\ w\to v\in\Edges}}\frac{1}{\Degree(w)}
        = \frac{1}{\Degree(v)} \sum_{\substack{u\to v\in\Edges \\ v\to u\in\Edges}}\frac{1}{\Degree(u)} \\
        \iff 
        &\frac{1}{\Degree(v)}
        = \frac{1}{\Degree(v)}
    \end{align}
    which is always true.
\end{proof}

\noindent For the following claims, we omit the proofs since the technique is
similar of that of the proof of \cref{thm:bi-directional-fps}.

\begin{definition}
    A \emph{carousel} is multipartite graph consisting of a partition of $\Vertices$
    into sets $S_1,\ldots,S_\ell$ such that $u\to v\in\Edges$ if and only if
    there exists an $i$ such that $u\in S_{i}$ and $v\in S_{i+1}$.
    We say that $S_{\ell+1}:= S_1$, $S_{0}:=S_\ell$, $S_{-1} := S_{\ell-1}$, etc.
\end{definition}

\begin{claim}\label{thm:carousel-fps}
    Suppose $\Graph=(\Vertices, \Edges)$ is a carousel
    with sets $S_1,\ldots,S_\ell$.
    Then for $v\in S_i$ we have that $\FixationProbabilityUnderNeutralEvolution{\Set{v}}\propto |S_{i-1}|^{-1}$.
\end{claim}
\begin{proof}
    We have
    \begin{align}
        |S_{i-1}|^{-1}\cdot\sum_{u\to v \in \Edges} \frac{1}{\OutDegree{u}}
        = |S_{i-1}|\cdot {|S_i|}^{-1} |S_{i-1}|^{-1}
        = {|S_i|}^{-1}
    \end{align}
    and
    \begin{align}
        \frac{1}{\OutDegree{v}}\sum_{v\to w\in\Edges} |S_{i}|^{-1}
        = |S_{i}|^{-1}.
    \end{align}
\end{proof}

\begin{definition}
    An \emph{$(s_1,\ldots,s_\ell)$-book} is a graph with vertices $b$ (beginning), $e$ (end),
    and $p^{i}_{j_i}$ for $i=1,\ldots, \ell$ and $j_i = 2,\ldots,s_{i-1}$.
    We say that $p^i_1 = p^i_{s_i+1} = b$ and $p^i_{s_i}=e$ for each $i$.
    We have $p^{i}_{j_i}\to p^{i}_{j_i+1}\in \Edges$ for all $j_i=1,\ldots,s_i$
\end{definition}
\begin{claim}\label{thm:book-fps}
    Suppose $\Graph=(\Vertices, \Edges)$ is a $(s_1,\ldots,s_\ell)$-book.
    Then
    \begin{enumerate}
        \item $\FixationProbabilityUnderNeutralEvolution{\Set{b}}\propto \ell$
        \item $\FixationProbabilityUnderNeutralEvolution{\Set{e}}\propto 1$
        \item $\FixationProbabilityUnderNeutralEvolution{\Set{p^i_2}}\propto \ell$
        for each $i$
        \item $\FixationProbabilityUnderNeutralEvolution{\Set{p^i_j}}\propto 1$
        for each $i$ and each $j > 1$.
    \end{enumerate}
\end{claim}
\begin{proof}
    This can be easily checked by plugging into the equations of \eqref{eq:lin-system}.
\end{proof}

\begin{definition}
    See \S 1.1.1 of \cite{galanis2017amplifiers} for the definition of a \emph{$(k, \ell, m)$-metafunnel}.
\end{definition}
\begin{claim}\label{thm:metafunnel-fps}
    Suppose $\Graph=(\Vertices, \Edges)$ is a $(k, \ell, m)$-metafunnel.
    Then
    \begin{enumerate}
        \item $\FixationProbabilityUnderNeutralEvolution{\Set{v}}\propto m^{1-i}$
        for each $v\in V_i$ such that $i\neq 1$
        \item $\FixationProbabilityUnderNeutralEvolution{\Set{v}}\propto \ell$
        for each $v\in V_1$.
    \end{enumerate}
\end{claim}

\begin{definition}
    See \S 1.1.2 of \cite{galanis2017amplifiers} for the definition of a \emph{$(k, \ell, m)$-superstar}.
\end{definition}
\begin{claim}\label{thm:superstar-fps}
    Suppose $\Graph=(\Vertices, \Edges)$ is a $(k, \ell, m)$-superstar.
    Then
    \begin{enumerate}
        \item $\FixationProbabilityUnderNeutralEvolution{\Set{v_{i,j}}}\propto 1$
        for each $i$ and each $j>1$
        \item $\FixationProbabilityUnderNeutralEvolution{\Set{v_{i,1}}}\propto 1/m$
        for each $i$
        \item $\FixationProbabilityUnderNeutralEvolution{\Set{v}}\propto \ell$
        for each $v\in R_i$ for each $i$
        \item $\FixationProbabilityUnderNeutralEvolution{\Set{v^*}}\propto 1$.
    \end{enumerate}
\end{claim}

\begin{corollary}
    Undirected graphs, carousels, books, metafunnels, and superstars are balanced.
\end{corollary}
\begin{proof}
    This follows from \cref{thm:balanced-graph-properties},
    and Claims \ref{thm:bi-directional-fps} to \ref{thm:superstar-fps}.
\end{proof}

\begin{definition}
    See \S 1.1.3 of \cite{galanis2017amplifiers} for the definition of a \emph{$(k, \ell, m)$-megastar}.
\end{definition}
\begin{claim}
    Suppose $\Graph=(\Vertices, \Edges)$ is a $(k, \ell, m)$-megastar.
    Then
    \begin{enumerate}
        \item $\FixationProbabilityUnderNeutralEvolution{\Set{v^*}}\propto m$
        \item $\FixationProbabilityUnderNeutralEvolution{\Set{v}}\propto \ell m$
        for each $v\in R_1\cup\cdots\cup R_\ell$
        \item $\FixationProbabilityUnderNeutralEvolution{\Set{v}}\propto m$
        for each $v\in K_1\cup\cdots\cup K_\ell$
        \item $\FixationProbabilityUnderNeutralEvolution{\Set{a_i}}\propto 1$
        for each $i\in[\ell]$.
    \end{enumerate}
\end{claim}

\begin{corollary}
     A $(k, \ell, m)$-megastar has $\FixationProbability_{\text{min}} \geq 1/h(N)$ for some polynomial
     $N$ so long as
     \begin{equation}
         1\cdot \ell + m\cdot (k\ell) + (\ell m)\cdot (\ell m) + m \leq h(N).
     \end{equation}
\end{corollary}
\noindent In particular, so long as $k$, $\ell$, and $m$ are each bounded above by some polynomial in $N$ then $\FixationProbability_{\text{min}}$ is bounded below by the inverse of a
polynomial.

\section{Quickly estimating fixation probabilities with provable confidence}\label{sec:fpras}

\begin{theorem}
    For $\Graph=(\Vertices,\Edges)$, $u\in\Vertices$, and $r\geq 1$,
    there is a fully polynomial randomised approximation scheme (FPRAS) for computing $\FixationProbability_r(\Graph, u)$ if there is some polynomial $h$ such that $\pmin \geq 1/h(N)$.
\end{theorem}

\begin{proof}
    What follows is a standard technique.
    We aim to approximate the fixation probability
    within a multiplicative factor of $\e>0$
    with probability $1-\nu$.
    Let $\nu_1,\nu_2\geq 0$ be constants and let $\nu:=\nu_1+\nu_2$.
    We first compute $\pmin$ and then $T_1$ in $\poly(N)$ time by Lemma
    \ref{lemma:fp-linear-system}.
    We run $s = \lceil 2(\ln 2/\nu_1)/(\e\FixationProbability_{\text{min}})^2\rceil$ independent simulations of the birth-death process on $\Graph$
    for at most $t = \lceil T_1s/\nu_2 \rceil$ steps.
    If any simulation does not reach absorption in the allocated time, we return some arbitrary number as the fixation probability.
    Otherwise, take $X_1,\ldots,X_s$ to be random $0$-$1$ indicator variables such that
    $X_i=1$ if and only if the simulation reaches fixation.
    We estimate $\hat{\FixationProbability}_r(\Graph,u) := \frac1s\sum_{i=1}^s X_i$ as the fixation probability.
    Then $\Pr[|\hat{\FixationProbability}_r(\Graph,u)-\FixationProbability_r(\Graph,u)|>\e\,\FixationProbability_r(\Graph,u)]\leq 2\exp(-\e^2s\,\FixationProbability_r(\Graph,u)^2/2)\leq \nu_1$ by a Chernoff bound.
    The probability that a simulation does not reach absorption in the allocated amount of
    time is at most $\nu_2/s$ by Markov's inequality. 
    By a union bound, the probability that there is some simulation that does not reach absorption in the allocated time is at most $\nu_2$.
    Thus the probability of error in this approximation algorithm is at most $\nu_1 + \nu_2 =\nu$.
    Both $s$ and $t$ are polynomial in $N$ and each step of the birth-death process can be computed
    in a constant amount of time.
\end{proof}

\section{Computer experiments}\label{sec:oriented}
\begin{figure}[h]
  \centering
   \includegraphics[width=\linewidth]{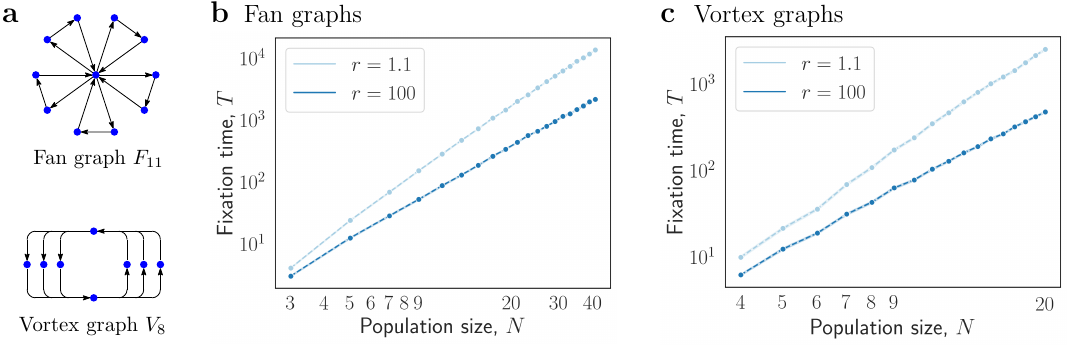}

\caption{
\textbf{Fixation time on slow oriented graphs.}
\textbf{a,} The Fan graph with $k$ blades has $N=2k+1$ nodes and $3k$ one-way edges (here $k=5$ which yields $N=11$). The Vortex graph with batch size $k$ has $N=2k+2$ nodes and $4k$ edges (here $k=3$ which yields $N=8$).
\textbf{b-c,} For both the Fan graphs and the Vortex graphs the fixation time scales roughly as $N^2$, both for $r=1.1$ and $r=100$. Each data point is an average over 100 simulations.
We note that the plots are on a log-log scale.
}
\label{fig:fans-loglog}
\end{figure}



\end{document}